\documentclass[11pt]{article}

\usepackage{graphicx,psfrag}

\usepackage{amssymb,amsmath,amsthm,enumerate}
\usepackage{fullpage}
\newcommand{\url}{}

\RequirePackage[colorlinks,citecolor=blue,urlcolor=blue]{hyperref}

\newtheorem{lemma}{Lemma}

\newtheorem{proposition}[lemma]{Proposition}
\newtheorem{definition}[lemma]{Definition}
\newtheorem{theorem}[lemma]{Theorem}

\newtheorem{corollary}[lemma]{Corollary}
\newtheorem*{Proposition*}{Proposition}

\newcommand{\dE}{\mathbb {E}}
\newcommand{\dP}{\mathbb{P}}
\newcommand{\dN}{\mathbb {N}}
\newcommand{\dR}{\mathbb {R}}
\newcommand{\dC}{\mathbb {C}}

\newcommand{\dZ}{\mathbb {Z}}

\newcommand{\cE}{\mathcal {E}}
\newcommand{\cF}{\mathcal {F}}
\newcommand{\cG}{\mathcal {G}}

\newcommand{\cX}{\mathcal {X}}
\newcommand{\cP}{\mathcal {P}}

\newcommand{\perc}{ \mathrm{perc}}
\newcommand{\supp}{ \mathrm{supp}}

\newcommand{\ABS}[1]{{{\left| #1 \right|}}} % |1|
\newcommand{\BRA}[1]{{{\left\{#1\right\}}}} % {1}
\newcommand{\dNRM}[1]{{{\left\| #1\right\|}}} % ||1||
\newcommand{\NRMTV}[1]{\dNRM{#1}_\mathrm{TV}} % Hilbert-Schmidt
\newcommand{\PAR}[1]{{{\left(#1\right)}}} % (1)

\newcommand{\1}{1\!\!{\sf I}}\newcommand{\IND}{\1}
\newcommand{\veps}{\varepsilon}

\newcommand{\PP}{\dP}
\newcommand{\Bin}{\mathrm{Bin}}

\newcommand{\GW}{\mathrm{GW}}
\newcommand{\UGW}{\mathrm{UGW}}

\newcommand{\DTV}{\mathrm{d_{TV}}}

\newcommand{\Lto}{\stackrel{L_1}{\longrightarrow}}
\newcommand{\Ltotwo}{\stackrel{L_2}{\longrightarrow}}

\title{%Extended states in nearly regular trees and quantum percolation in finite graphs
On quantum percolation in finite regular graphs}
\author{Charles Bordenave\thanks{Research partially supported by ANR-11-JS02-005-01}}

\begin{document}
\maketitle

\begin{abstract}
The aim of this paper is twofold. First, we study eigenvalues and eigenvectors of the adjacency matrix of a bond percolation graph when the base graph is finite and well approximated locally by an infinite regular graph. We relate quantitatively the empirical measure of the eigenvalues and the delocalization of the eigenvectors to the spectrum of the adjacency operator of the percolation on the infinite graph. Secondly, we prove that percolation on an infinite regular tree with degree at least $3$ preserves the existence of an absolutely continuous spectrum if the removal probability is small enough. These two results are notably relevant for bond percolation on a uniformly sampled regular graph or a Cayley graph with large girth.
\end{abstract}

\section{Introduction}

In the seminal work \cite{PhysRev.109.1492},  Anderson has studied the transport properties of a quantum particle on a regular lattice in the presence of random impurities. Shortly after, De Gennes, Lafore and Millot \cite{deGennes1959a,deGennes1959b} have proposed to study the transport properties on a randomly disordered lattice. The latter is now usually referred as quantum percolation and only results on the density of states are currently available \cite{PhysRevB.6.3598,MR869300,MR2148799,BSV}, see notably \cite{MR2732070,MR3051702} for survey and references. Mathematically, it amounts to study the regularity of the spectral measures of Laplacian-type operators of the disordered lattice. Since the landmark result of Klein \cite{klein}, perturbation methods have been used to study random operators on infinite regular trees (Bethe lattice). The spectrum of the Laplacian operator of a nearly-regular Galton-Watson tree without leaves has notably been studied recently by Keller \cite{MR2994759}. In parallel, Kottos and Smilansky \cite{MR1694731} have suggested that spectral statistics of some finite graphs are in good agreement with random matrix theory and the predictions of quantum chaos. In \cite{MR1903150,MR2768284}, Terras  has also discussed the connections between finite quantum chaos, spectrum of graphs and random matrix theory. On large finite regular graphs, local spectral distribution and delocalization of eigenvectors have notably been studied in \cite{MR1397154,MR3025715,MR3038543,MR2999215,Nalini1,geisinger}.

In this paper, in the spirit of De Gennes, Lafore and Millot, we study eigenvalues and eigenvectors of the adjacency matrix of finite percolation graphs. More precisely, we consider a large graph $G$ which is well approximated locally by an infinite regular graph, say $\Gamma$. We keep each edge of $G$ independently with probability $p$, remove it otherwise, and consider the adjacency matrix of the corresponding randomly diluted graph, denoted by $\perc(G,p)$. We relate quantitatively some notion of regularity of the spectral measure and the delocalization of the eigenvectors of $\perc(G,p)$ to the spectrum of the adjacency operator of the percolation on $\Gamma$, $\perc(\Gamma,p)$. These finite volumes corrections are stated in a general framework and have nearly the same order than the recent results on $d$-regular graphs and $p=1$, \cite{MR3038543,MR2999215,Nalini1,geisinger}. Using a perturbation technique, we also complement the result of Keller \cite{MR2994759} on supercritical Galton-Watson trees to the situation where the tree may have leaves.  Compared to the above mentioned results, an important new difficulty in quantum percolation is the concomitant presence of point and continuous spectrum, see \cite{PhysRevB.6.3598,MR869300} or \cite[\S 3.2]{coursSRG}.

This paper is organized as follows. The remainder of the introduction presents the main definitions and the main results.  Section \ref{sec:GWT} analyses the spectrum of Galton-Watson trees whose offspring distribution is close to deterministic. Section \ref{sec:DRB} presents basic resolvent bounds. Finally, in Sections \ref{sec:rate}-\ref{sec:deloc} we apply these bounds to finite percolation graphs.

\subsection{Adjacency operator and spectral measures}

Let $G = (V,E)$ be a locally finite (undirected) graph, that is a simple graph such that all vertices have a finite degree.  We may consider the Hilbert space 
$$\ell^2(V)=\left\{\psi \colon V\to \dC,\; 
  \sum_{x\in V}|\psi(x)|^2<\infty\right\},$$  with inner product $ \langle \phi , \psi \rangle=\sum_{x\in
  V}\overline{\phi(x)} \psi(x)$. 
Denote by $\ell^2_{0}(V) \subseteq \ell^2(V)$ the dense subspace of finitely
supported functions, and by $( e_x, x\in V)$ the canonical
orthonormal basis of $\ell^2(V)$, i.e. $e_x$ is the
coordinate function $y\in V\mapsto \IND(x=y)$. The adjacency operator $A$ of $G$ is the linear
operator over $\ell^2(V)$ whose domain is $\ell^2_0 (V)$ and whose action on the 
basis vector $e_x,x\in V$ is:
\begin{eqnarray*}
A e_x = \sum_{y:xy\in E}e_y. 
\end{eqnarray*}
Note that $A e_x\in\ell^2(V)$ since $G$ is locally finite. 
Moreover, for all $x,y\in V$, 
\begin{eqnarray*}
\langle A e_x, e_y\rangle = \IND\{xy\in E\} = \langle A e_y, e_x\rangle.
\end{eqnarray*}
Hence, the operator $A$ is symmetric, and we may ask about its (essential) self-adjointness, that is the self-adjointness of its closure. If the degrees of vertices of $G$ are uniformly bounded then $A$ is a bounded operator and hence self-adjoint on $\ell^2 (V)$. We will pay a special attention to adjacency operators of tree. A sufficient condition for self-adjointness was given in \cite[Proposition 3]{MR2789584}.

Recall that if $A$ is essentially self-adjoint then the spectral measure at vector $e_x$ (or at vertex $x \in V$) is well-defined. It is the unique probability measure on $\dR$, denoted by $\mu_G^{e_x}$, such that for all integer $k \geq 0$, 
\begin{equation}\label{eq:defmuGex}
\int x^k d \mu_{G}^{e_x} = \langle e_x, A^k e_x \rangle. 
\end{equation}
Note that the right hand side is the number of closed paths in $G$ of length $k$ starting at $x$. We will say that $A$ has non-trivial absolutely continuous spectrum if there exists $x \in V$ such that $\mu_G^{e_x}$  has an absolutely continuous part of positive mass. Also, $A$ has a purely absolutely continuous spectrum on an interval $(a,b)$ if for all $x \in V$,  $\mu_G^{e_x}$ is absolutely continuous on $(a,b)$.

If $G$ is a finite with $|V| = n$ vertices, we define classically the spectral measure as 
$$
\mu_G = \frac 1 {n} \sum_{k=1} ^n \delta_{\lambda_k}, 
$$
where $\lambda_1, \ldots , \lambda_n$ are the real eigenvalues of $A$. It is straightforward to check that $\mu_G$ can also be written as the spatial average of the spectral measures at the vertices :
$$
\mu_G = \frac 1 {|V|} \sum_{v \in V} \mu_G^{e_v}.
$$
Motivated by the Benjamini-Schramm local graph topology, see \cite{aldouslyons}, we will be interested by random rooted graphs $(G,o)$, random graphs with a distinguished vertex $o \in V$, the root. In which case, the expected spectral measure $\dE \mu_G^ {e_o}$ will play an important role. When the law of $(G,o)$ is unimodular (that is satisfies a specific mass transport principle), $\dE \mu_G^ {e_o}$ can be interpreted as a density of states. We refer to the introduction of \cite{BSV} and \cite{coursSRG} for more details. 

If $G = (V,E)$ is a graph, we will denote by $\perc (G,p)$ the random graph with vertex set $V$ and edge set $E' \subset E$ obtained by keeping each edge of $E$ independently with probability $p$ and removing it otherwise.

To motivate the sequel, we now briefly argue that expected spectral measures of percolation graphs has typically a dense set of atoms on its support. Take $0 < p < 1$ and let $G = \perc ( \Gamma, p)$ be the percolation graph of some infinite graph $\Gamma$ with uniformly bounded degrees (for example $\Gamma$ is the lattice $\dZ^d$ or the infinite $d$-regular tree), then $G$ will have finite connected components with probability one. The spectral measures at vertices belonging to these finite connected components will be pure point. More importantly, the spectral measure of some vertices on infinite connected components will also have non-trivial atomic parts. This is notably due to the presence of finite pending subgraphs, indeed, it is not hard to build localized eigenvectors associated to eigenvalues of the adjacency matrix of these finite pending subgraphs, for a detailed argument see \cite{PhysRevB.6.3598,MR869300} or \cite[\S 3.2]{coursSRG}.

\subsection{Extended states in Galton-Watson trees}

Let $P = (P_k)_{k \geq 0} \in \cP ( \dZ_+)$ be a probability distribution on non-negative integers. A Galton-Watson tree with offspring distribution $P$ ($\GW(P)$ tree for short) is the random rooted $(T,o)$ defined as follows. Let $\dN_f= \cup_{k \geq 0} \dN^{k}$ with $\dN^0 = \{ o \}$ be the set of finite sequences of integers and let $(N_x)_{ x \in \dN_f}$ be independent  variables with common distribution $P$.  The vertex set $V$ of $T$ is the subset of $\dN^f$ obtained iteratively as follows: the offspring of $x = (i_1, \cdots, i_k) \in \dN^k \cap V$ are $V_x = \{ (i_1, \cdots , i_k, \ell ) , 1 \leq \ell \leq N_x\}$.  Proposition 7 in \cite{MR2789584} asserts that if $\dE N_o <  \infty$ and $A$ is the adjacency operator of $T$  then with probability one, $A$ is essentially self-adjoint.

In the specific case where $P = \delta_q$ is a Dirac mass at $q$, then $T$ is the infinite $q$-ary tree. It is not hard to check (see forthcoming Section \ref{sec:GWT}) that, in this case, $\mu_{T}^{e_o}$ is the Wigner semicircle distribution with radius $2 \sqrt q$. More precisely, $\mu_{T}^{e_o}$ has density on $[-2 \sqrt q, 2 \sqrt q]$ given by
$$
f_q (\lambda) = \frac{1}{2 \pi  q} \sqrt{ 4q - \lambda^2}.  
$$
Our first result asserts that if $P$ is close enough to $\delta_q$ then the adjacency operator of $T$ has an absolutely continuous part. For $p \geq 1$, if $N$ has distribution $P$, the Wasserstein $L^p$-distance to the Dirac mass $\delta_q$ is given by 
$$
W_p ( P , \delta_q ) = \dE |N - q |^p = \sum_{k = 0}^\infty | k - q |^p P ( k). 
$$

\begin{theorem}\label{th:GW}
Let $A$ be the adjacency operator of $T$, a $\GW(P)$ tree. Let $q \geq 2$ be an integer. There exists $\veps  = \veps(q) >0$ such that if $W_1 ( P , \delta_q) < \veps$, then $A$ has a non-trivial absolutely  continuous spectrum with positive probability. Moreover,  if $f$ denotes the density of the absolutely continuous part of the spectral measure at the root $\mu_T ^{e_o}$ of $A$,
\begin{equation}\label{eq:convfac}
\lim_{ P \Lto\delta_q } \int  \dE | f (\lambda)  - f_{q} (\lambda) | d \lambda = 0. 
\end{equation}
\end{theorem}

Theorem \ref{th:GW} will be proved by using a technique first developed in Aizenman, Sims and Warzel \cite{MR3055759}. For any $p > 1$ and $q \geq 2$, Keller \cite{MR2994759} has proved that there exists some $\veps'  = \veps'(p,q)> 0$ such that if $W_{p} ( P , \delta_q) < \veps'$ and $P(0) = 0$ then $A$  has absolutely  continuous spectrum with probability one. In particular, Theorem \ref{th:GW} complements, Keller's result when $P(0) \ne 0$. Theorem \ref{th:GW}  has the following corollary on the density of states.

\begin{corollary}\label{cor:GW}With the notation of Theorem \ref{th:GW}, 
if $W_{1} ( P, \delta_q ) < \veps$, the expected spectral measure $\dE \mu_T^{e_o}$ has a non-trivial absolutely continuous  part $\bar f(\lambda) d\lambda$ and 
\begin{equation}\label{eq:convdos}
\lim_{ P \Lto\delta_q } \int | \bar f (\lambda)  - f_{q} (\lambda) | d \lambda = 0. 
\end{equation}
\end{corollary}

As already mentioned, unimodular random rooted graphs plays a central role in Benjamini-Schramm local graph topology. Assume that $P$ has a positive and finite first moment. The unimodular Galton-Watson tree with degree distribution $P$ ($\UGW(P)$ tree for short) is the random rooted $(T,o)$ defined as above, where $N_o$ has distribution $P$ and for all other $x \in \dN^f \backslash \{ o\}$, $N_x$ are independent with  common distribution  $\widehat P$ defined by 
$$
\widehat P (k)  = \frac{ (k+1) P(k+1)}{ \sum_{\ell} \ell P(\ell)}. 
$$
 As its name suggests, the distribution $\UGW(P)$  is unimodular, see \cite{aldouslyons}. Also, this distribution is the Benjamini-Schramm limit of numerous graph sequences. For example, if $q \geq 1$ and $P = \delta_{q+1}$, then $\hat P = \delta_q$ and $\UGW(\delta_{q+1})$ is a Dirac mass at the infinite $(q+1)$-regular tree. In this case, $\mu_{T}^{e_o}$ is the Kesten-McKay distribution, it has density on $[-2 \sqrt q, 2 \sqrt q]$ given by
$$
\check f_q (\lambda) = \frac{(q+1)}{2 \pi} \frac{ \sqrt{ 4q - \lambda^2}}{(q+1)^2 - \lambda^2}.  
$$Our next result adapts the above statements to unimodular Galton-Watson trees. 

\begin{theorem}\label{th:UGW}
Let $A$ be the adjacency operator of $T$,  a $\UGW(P)$ tree. Let $q \geq 2$ be an integer. There exists $\veps = \veps (q)>0$ such that if $W_{2} ( P, \delta_{q+1} ) < \veps$, then $A$ has a non-trivial absolutely  continuous spectrum with positive probability. Moreover, the conclusions \eqref{eq:convfac} of Theorem \ref{th:GW} and \eqref{eq:convdos} of Corollary \ref{cor:GW} hold with $f_q$ replaced by $\check f_q$ and $P \Lto \delta_q$ by $P \Ltotwo \delta_{q+1}$. 
\end{theorem}

We may apply the above theorem to bond percolation on the infinite $(q+1)$-regular tree, say $T_{q}$. Then, the connected component of the root in $\perc (T_{q}, p)$ has distribution $\UGW ( \Bin (  q+1, p ) )$. Recall that this connected component is infinite with positive probability if and only if $p q > 1$. If $A$ is the adjacency operator of a $\UGW( \Bin (  q+1, p ) )$ tree, in \cite{BSV}, it is proved that the density of states $\dE \mu_{\perc (T_{q}, p)}^{e_o}$ has a non-trivial continuous part if and only if $p q > 1$.  We may define the quantum percolation threshold, 
\begin{equation*}\label{eq:defpq}
p_q = \sup \{ p \geq 0  : A \hbox{ has no absolutely continuous spectrum with probability one} \}. 
\end{equation*}
By unimodularity \cite[Lemma 2.3]{aldouslyons}, $p_q$ is also equal to 
$$
p_q = \sup \{ p \geq 0  : \mu_{\perc (T_{q}, p)} ^{e_o} \hbox{ has a trivial absolutely continuous part with probability one} \}. 
$$
We may also define a mean quantum percolation threshold,
\begin{equation*}
\bar p_q = \sup \{ p \geq 0  : \dE \mu^{e_o}_{\perc (T_{q}, p)} \hbox{ has a  trivial  absolutely continuous part} \}. 
\end{equation*}
From what precedes,
$
1  / q \leq \bar p_q \leq p_q \leq 1.
$
As a corollary of Theorem \ref{th:UGW}, we find

\begin{corollary}\label{cor:UGW}
For any integer $q \geq 2$, we have $p_q < 1.$ 
\end{corollary}

Note that due to the lack of monotonicity, it is not clear whether $p_q = p_q^*$ where 
$
p^*_q = \inf \{ p \leq 1  : \mu_{\perc (T_{q}, p)} ^{e_o} \hbox{ has a non-trivial absolutely continuous part with positive probability} \}. 
$ It is also unknown if the strict inequality $\bar p_q < p_q$ holds or not.

\subsection{Rates of convergence  in percolation graphs}

We now give quantitative finite size corrections on linear functions of the eigenvalues of percolation graphs. Let $\Gamma = (V(\Gamma),E(\Gamma))$ be a vertex transitive graph and $G$ a finite graph on $n$ vertices. For integer $h \geq 1$ and $v \in V(G)$, we denote by $(G,v)_h$ the subgraph of $G$ spanned by the vertices which are at distance at most $h$ from $v$.  Also, $B_\Gamma(h)$ is the number of vertices of $G$ such that $(G,v)_h$ is not isomorphic to $(\Gamma,o)_h$ where $o \in V(\Gamma)$.

For example, let $q \geq 2$ be an integer and assume further that $G$ is a $(q+1)$-regular graph. Then $B_{T_q} (h)$ is the number of vertices $v$ in $V(G)$ such that $(G,v)_h$ is not a tree. Observe that if $G$ has girth $g$ (length of the shortest cycle), then $B_{T_q} (h) = 0$ for $h < g/2$. Also, if $G$ is a uniformly sampled $(q+1)$-regular graph on $n$ vertices, then, for any $0 < \alpha <1$, with probability tending to $1$ as $n \to \infty$, $B_{T_q} (h) = n ^{ \alpha + o(1)}$ where $h = \lfloor ( \alpha \log n ) / ( 2 \log q ) \rfloor$. This follows from known asymptotic on the number of cycles in random regular graphs, see \cite{MR3025715,MR2097332}.

If $\varphi : \dR \to \dR$ has its derivative $\partial \varphi$ in $L^1(\dR)$, we set $\NRMTV{\varphi} = \int |\partial \varphi (x)| dx$. The next statement will be a consequence of Jackson's approximation theorem. 

\begin{theorem}\label{th:rate}
Let $G$ be a graph with $n$ vertices and maximal degree $d$. For $p \in [0,1]$, let $\mu = \dE \mu^{e_o} _{ \perc ( \Gamma, p)}$ and let $\varphi$ be a $C^k$-function with $1 \leq k \leq 2h$. We have 
$$
\ABS{ \int \varphi d \mu_{\perc (G,p)}  - \int \varphi d \mu  } \leq t + 2 \frac {B_\Gamma(h)}{n} \| \varphi \|_{\infty} + 2 \PAR{ \frac \pi 2  d }^k \frac{ (2h - k +1)!}{(2h +1) ! }  \| \partial^{(k)} \varphi \|_{\infty},
$$ 
with probability at least $1 - 2 \exp ( - n t ^2 / ( 8 \NRMTV{\varphi}^2 )  )$.  
\end{theorem}

To understand better the above statement, we can apply it to the Cauchy-Stieltjes transform $\varphi_z (x ) = 1 / ( x - z)$ which we will denote by
$$
g_\mu ( z) = \int \frac{d \mu(\lambda) }{\lambda -z}.
$$
 Rougly speaking, if $h B_{\Gamma} (h) = o (n)$, the Cauchy-Stieltjes transform converges as soon as $\Im (z) \geq C ( \log h ) / h$. More precisely, we will obtain for example the following corollary.

\begin{corollary}\label{cor:rate}
Let $G$ be a graph with $n$ vertices and maximal degree $d \geq 2$. For $p \in [0,1]$, let $\mu = \dE \mu^{e_o} _{ \perc ( \Gamma, p)}$ and
\begin{equation}\label{eq:defdmin}
\delta \geq   \frac{h B_\Gamma (h)}{n} \vee \frac 1 h . 
\end{equation}
Then, for any $z \in \dC$ with $\Im (z)  \geq 20 d \log (2h) / h$, 
$$
\ABS{ g_{ \mu_{\perc (G,p)}} (z) - g_{\mu} (z)  } \leq  \delta,
$$ 
with probability at least $1 - 2 \exp ( - n  \delta^2 /  h^2  )$.  
\end{corollary}

If $p=1$ and $G$ is a $d$-regular graph, the above corollary recovers, up to the $ O ( \log h)$ factor for the lower bound on $\Im (z)$, statements in \cite{MR3025715,geisinger}, see also \cite{MR3038543,Nalini1}. It would be very interesting to extend Corollary \ref{cor:rate} to some $z \in \dC$ with $\Re(z)$ in the support of $\mu$ and $\Im (z) = o ( 1/ h)$.  Also in the bound \eqref{eq:defdmin}, the term $1/h$ could be replaced by $1/h^k$ for any $k \geq 1$ by increasing suitably the constant $20$. This would however not change much for the applications that will follow.

\subsection{Regularity of the spectral measure in percolation graphs}

As above, $\Gamma$ is a vertex transitive graph. The next statement asserts that if $\mu_{\perc (\Gamma,p)}$ has an absolutely continuous part and $B_\Gamma (h) = o(n)$  for some $h \gg 1$, then $\mu_{\perc(G,p)}$ will also have some regularity property on intervals of scale $1/h$. The Lebesgue measure on $\dR$ is denoted by $\ell$. 
-
\begin{theorem}\label{th:perclocallaw}
Let $G$ be a graph with $n$ vertices, maximal degree $d \geq 2$ and $\delta$ as in \eqref{eq:defdmin}. For some $p \in (0,1]$, assume that $\mu = \dE \mu^{e_o}_{\perc ( \Gamma,p)}$ has an absolutely continuous part. Then, for any $\veps >0$, there exist positive constants $c_0,c_1$ and a deterministic closed set  with $\ell( K) > 0$ such that with probability at least $1 -  \delta^{-1} h^2  \exp ( -  n \delta^2 / h^2)$,  the following holds:
\begin{enumerate}[(i)]
\item If $\delta \leq c_0$, then for any $\lambda \in K$ and interval $I = ( \lambda - \eta , \lambda + \eta )$ with $ \eta \geq \eta_{\min}= c_1( \log h)   / h$,
$$
 \frac{\mu_{\perc(G,p)}( I ) }{ \ell ( I) } \leq c_1 \quad \hbox{ and }  \quad  \frac{\mu_{\perc(G,p)}( \bar I ) }{ \ell ( I) }  \geq c_0 \frac{ \eta_{\min}}{\eta}.
$$
\item
If $\mu_s(\dR)$ is the total mass of the singular part of $\mu$, we have $\mu_{\perc(G,p)} (K^c ) \leq 2 \pi (\mu_s ( \dR) + \veps + \delta)$.
\end{enumerate}
\end{theorem}

In the proof, the constants $c_0, c_1$ and the set $K$ will depend on the absolutely continuous part of $\mu$ and $\veps$ in a rather straightforward manner.
Note also that due to the $2 \pi$ factor, statement (ii) is only useful when $\mu_s ( \dR)$ is small enough. From Corollary \ref{cor:UGW}, it is the case for example in $\perc ( T_q  , p)$ when $p$ is close to $1$. 

\subsection{Weak delocalization in percolation graphs}

We now turn to statements on delocalization of eigenvectors of $\perc ( G , p)$ when the adjacency operator of $\perc ( \Gamma, p )$ has a non-trivial continuous spectrum with positive probability. We will use a rather weak notion of delocalization in the underlying canonical basis.

\begin{definition}
Let $(\rho,\veps) \in [0,1]^2$. A unit vector $\psi \in \dC^n$ is $(\rho,\veps)$-delocalized if there exists $S \subset [n]$ such that $\sum_{i \in S} |\psi(i)|^2 \geq \rho$ and for all $i \in S$, $|\psi(i)| \leq \veps$. 
\end{definition}

We also introduce some volumetric parameters of $G$.  For $v \in V$, we set 
\begin{equation}\label{eq:NhMh}
N_h (G,v) = | V ( (G,v)_h ) | \quad  \hbox{ and  }\quad  M_h (G) = \PAR{\frac 1 {|V|} \sum_{v \in V} N^2_h ( G,v)}^{1/2}. 
\end{equation}
In words, $N_h ( G,v)$ is the number of vertices at graph distance at most $h$ from $v$ and $M_h (G)$ is its quadratic average. Observe that if $G$ has maximal degree $d$, then $N_h ( G,v)\leq d ( d-1) ^{h-1}$. 

\begin{theorem}\label{th:percdeloc}
Let $G$ be a graph with $n$ vertices, maximal degree $d \geq 2$ and  $\delta$ as in \eqref{eq:defdmin}. For some $p \in (0,1]$, assume that $\mu^{e_o}_{\perc ( \Gamma,p)}$ has an absolutely continuous part with positive probability. Consider an orthogonal basis of eigenvectors of the adjacency matrix of $\perc(G,p)$. Then, for any $\veps >0$, the following holds for some positive constants $c_0,c_1, \alpha, \rho$ (depending on $\veps,d,\perc ( \Gamma,p)$). With probability at least $1 -  \delta^{-1} h^2  \exp ( -  n \delta^2 / (2 h^2 M_h(G)^2) )$,  we have \begin{enumerate}[(i)]
\item If $\delta \leq c_0$, at least $\alpha n$ eigenvectors of $\perc(G,p)$ are $\PAR{\rho, c_1  \sqrt{\frac{ \log h }{ h}}  }$-delocalized.

\item
If $\bar \mu_s ( \dR)$ is the expected mass of the singular part of $\mu^{e_o}_{\perc ( \Gamma,p)}$, then we can take $\alpha = \rho  = 1 - \sqrt{ 4 \pi (\bar \mu_s ( \dR) + \veps + \delta)}$.
\end{enumerate}
\end{theorem}

Again, the dependency of the constants in terms of the distribution of $\mu_{\perc(\Gamma,p)} ^{e_o}$ will be explicit. Our delocalization statement is weaker than other delocalization results obtained in \cite{MR3038543,Nalini1} on tree-like $d$-regular graphs and $p=1$. Even, for this simpler class of graphs, it is an open problem to prove $(\rho, \veps)$-delocalization with $\veps = o (1/ \sqrt h)$. From Theorem \ref{th:UGW}, statement (ii) could be applied to $\perc(T_d,p)$ and $p$ close to $1$. The proof of Theorem \ref{th:percdeloc} will also rely on resolvent methods inspired by  \cite{ESY10}.

The proof of Theorems \ref{th:perclocallaw} and Theorem \ref{th:percdeloc} will rely on resolvent methods introduced notably in the context of random matrices by Erd\H{o}s, Schlein and Yau in \cite{ESY10}.

%Finally, without major modifications, our statements can be generalized beyond percolation graphs, see forthcoming Subsection \ref{subsec:bpg}. 

\section{Spectrum of Galton-Watson trees}

\label{sec:GWT}
\subsection{Resolvent operator}

Let $P = (P_k)_{k \geq 0} \in \cP ( \dZ_+)$ with finite first moment and let $T$ be a $\GW(P)$ tree. As already pointed, with probability one, the adjacency operator $A$ is essentially self-adjoint. We may thus define its resolvent operator for $z \in \dC_+ = \{  z \in \dC : \Im (z) >0\}$, as
$$
G(z) = ( A- z I) ^{-1}. 
$$

For $x \in V$, we introduce $T_x$ the subtree rooted at $x$ spanned by the vertices whose common ancestor is $x$. The trees $T_{x}, x \in V_o$, are, given $N_o$, independent with common distribution $\GW(P)$. We denote by $A_x$ the adjacency operator of $T_x$ and set for $z \in \dC_+$, 
$$G_x(z) = \langle e_x , ( A_x- z I) ^{-1}e _x\rangle. $$ 
A well-known consequence of Schur's complement formula is, for all $z \in \dC_+$, 
\begin{equation}\label{eq:schur}
G_o (z)= - \PAR{ z +   \sum_{x \in V_o} G_x (z) }^{-1},
\end{equation}
see e.g. Klein \cite[Proposition 2.1]{klein} or \cite{MR2724665,MR2789584}. Since $G_x$ and $G_o$ have the same distribution, it follows that $G_o$ satisfies a recursive distribution equation which we are going to study in the regime $P$ close to a Dirac mass and $\Im(z) \to 0$ in the next subsections.

From \eqref{eq:defmuGex}, $G_x$ is the Cauchy-Stieltjes transform of the random probability measure of $\dR$, $\mu_x = \mu^{e_x}_{T_x}$, i.e. 
$$
G_x (z) = g_{\mu_x} (z) = \int \frac{1}{\lambda - z} d \mu_x (\lambda).
$$
Almost everywhere, the limit 
$$
G_x ( \lambda + i0) := \lim_{\eta \downarrow 0} G_x(\lambda+ i\eta)
$$
exists and  the density of $\mu_x$ at $\lambda \in \dR$ is given by 
$
 \Im G_x ( \lambda + i0) / \pi 
$,
see   e.g. Simon \cite[Chapter 11]{MR2154153}.

\subsection{Skeleton of a Galton-Watson tree}

We introduce the subset $S \subset V$ of vertices $x \in V$ such that $T_x$ is an infinite tree. Let  
$\pi_e$ be the probability that $o \notin S$, $\pi_e$ is the extinction probability of $T$ and it is the smallest root of the equation
$$
x = \varphi(x), 
$$
where 
$$
\varphi(x ) = \dE \PAR{ x^{N_o}} = \sum_{k=0}^{\infty} P_k x^k 
$$
is the moment generating function of $P$. Also, if $P \ne \delta_1$, the condition $m_1 > 1$ is equivalent to $\pi_e< 1$. 

Let $N_s$ and $N_e$ be the number of offspring of the root in $S$ and not in $S$.  The pair $(N_s,N_e)$ has the same distribution than $( \sum_{i=1} ^N ( 1 -\veps_i) , \sum_{i=1} ^N \veps_i)$ where $N$ has distribution $P$ and is independent of the $(\veps_i)_{i \geq 1}$ an i.i.d. sequence of Bernoulli variables with $\dP ( \veps_i = 1) = \pi_e = 1 - \dP ( \veps_i = 0)$. Moreover, conditioned on the root is in $S$, $(N_s,N_e)$ is conditioned on $N_s \geq 1$. In the sequel $(N'_s, N'_e)$ will denote a pair of random variables with  distribution $(N_s,N_e)$ conditioned on $N_s \geq 1$. In particular, the moment generating function of $(N'_s,N'_e)$ is given by 
\begin{equation}\label{eq:varphis}
\varphi_{s,e} (x,y) = \dE\BRA{ x^{N'_s} y ^{N'_e}  }  =   \dE\BRA{ x^{N_s} y ^{N_e}  | N_s \geq 1} = \frac{ \varphi ( (1 - \pi_e ) x + \pi_e y ) - \varphi( \pi_e y)  } { 1 - \varphi( \pi_e y) }. 
\end{equation}
Similarly, given $o \notin S$, $(N_s,N_e)$ is conditioned on $N_s = 0$.  Then, we find easily that the moment generating function of $N_e$ given $o \notin S$ is 
\begin{equation}\label{eq:varphie}
\varphi_e (x) = \frac{ \varphi ( \pi_e x ) } {   \pi_e}. 
\end{equation}
(For more details see Athreya and Ney \cite[Section I.12]{MR0373040}, Durrett \cite[Section 2.1]{MR2656427}).

We deduce from \eqref{eq:schur} that the variable $G^s_o$ defined as the law of $G_o$ conditioned on $o \in S$, satisfies the recursive distribution equation, 
\begin{equation}
\label{eq:schur2}
G^s_o (z) \stackrel{d}{=} - \PAR{ z + \sum_{x=1}^{N'_s}  G^s_x (z) + V (z) }^{-1},  
\end{equation}
where $G^s_x$ are independent copies of $G^s_o$, independent of $(N'_s,V(z) ) $ defined by 
\begin{equation}\label{eq:defV}
V  (z) =  \sum_{x = 1} ^{N'_e} G^e_{x} (z), 
\end{equation}
and $G^e_x$ are independent copies of $G_o$ given $o \notin S$ and are independent of $(N'_s,N'_e)$ with moment generating function given by \eqref{eq:varphis}.

Now if $P$ is close to $\delta_q$ with $q \geq 2$, then the central idea is to interpret  \eqref{eq:schur2} has a stochastic perturbation of the deterministic equation, $g(z) \in \dC_+$ and  
\begin{equation}\label{eq:defgsc}
g (z)  =  - \PAR{ z + q g(z) }^{-1},  
\end{equation}
which characterizes the Cauchy-Stieltjes transform of the semicircle distribution with radius $2 \sqrt q$. This is the objective of the next subsection.  We will use that, as a function of $P$, the extinction probability is weakly continuous at any $P \ne \delta_1$.
\begin{lemma}\label{le:contextinc}
The map $P \mapsto \pi_e(P)$ from $\cP ( \dZ_+) $ to $[0,1]$ is continuous for the weak convergence at any $P \ne \delta_1$ 
\end{lemma}

\begin{proof}
Take $P \ne \delta_1$. Fix a sequence of probability measures $P_n$ converging weakly to $P$. We set $\pi_n = \pi_e (P_n)$, $\pi_\infty= \pi_e (P)$, and we should prove that $\pi_n \to \pi_\infty$. We denote by $\varphi_n$ and $\varphi$ the generating functions of $P_n$ and $P$. We have the uniform convergence 
\begin{equation}\label{eq:unifgfu}
\max_{ x \in [0, 1] } | \varphi_n (x) - \varphi (x) | \to 0,
\end{equation}
(see Kallenberg \cite[Theorem 4.3]{MR1876169}).
We first prove that $\liminf_n \pi_n \geq \pi_\infty$. Consider a subsequence of $\pi_{n'}$ converging to $\pi' \in [0,1]$. Using \eqref{eq:unifgfu} and the continuity of $\varphi$,  we find that 
$$
0  = \varphi_{n'} ( \pi_{n'} ) -   \pi_{n'}  = \varphi ( \pi_{n'} ) - \pi_{n'} + o(1) = \varphi ( \pi') - \pi'  + o(1).
$$
In particular $\varphi (\pi') = \pi' $ and $\pi' \geq \pi_\infty$ from the definition of $\pi_\infty$.

To conclude of the proof of the lemma, it remains to check that $\limsup_n \pi_n \leq \pi_\infty$. We may assume that $\pi_\infty < 1$ otherwise there is nothing to prove. In particular, since $P \ne \delta_1$, we have $m_1 > 1$ and $P \ne P_0 \delta_0 + P_1\delta_1$. Fix any $x \ \in ( \pi_\infty , 1 )$,  the function $\varphi$ is strictly convex and $\varphi (x ) - x  < 0$. From \eqref{eq:unifgfu}, we deduce that for all $n$ large enough,  $\varphi_n (x) - x < 0$.  Hence, $\pi_n < x$. Since $x$ may be arbitrarily close to $\pi_\infty$, we get $\limsup_n \pi_n \leq \pi_\infty$. 
\end{proof}

We will use the straightforward consequence of Lemma \ref{le:contextinc}. Recall that $q \geq 2$.

\begin{corollary}\label{cor:contextinc}
There exists $\veps > 0$ such that if $W_1 (P, \delta_q) \leq \veps$ then $\dE [ N'_s + N'_e] \leq 2 q$.  
\end{corollary}

\begin{proof}
By Lemma \ref{le:contextinc}, if $\veps$ is small enough, we have $\pi_e(P) \geq 3/4$. Then  $\dE [ N_s + N_e | N_s \geq 1] \leq  (4/3) \dE [ N_s + N_e ] \leq 2 q$ if $\veps$ is small enough.
\end{proof}

\subsection{Convergence of the resolvent in the upper half-plane}

 We first check that the resolvent converges when $\Im(z) >0$. The total variation distance between two probability measures  $P$ and $Q$ on $\dZ_+$ is classically defined as 
$$
\DTV ( P , Q) = \frac 1 2 \sum_{k = 0}^\infty \ABS{ P(k) - Q (k)}. 
$$
The total variation distance is a metric for the weak convergence on discrete spaces. We have $\DTV(P,\delta_q) = \dP ( N \ne q)$ if $N$ has distribution $P$. 

\begin{lemma}\label{le:ASS}
For any $z \in \dC_+$, if $\DTV( P,\delta_q) \to 0$ then $G_o (z)$ and $G_o ^s(z)$ converge weakly to $g (z)$.
\end{lemma}
\begin{proof}
Set $\eta = \Im (z) >0$. In \eqref{eq:schur}, $N = |V_o|$ is independent of $G_x$, thus we find 
\begin{eqnarray*}
\dE | G_o (z) - g(z) |&  \leq &  \eta^{-2} \dE \BRA{\ABS{ \sum_{x=1}^q G_x (z) - q g(z) } \IND_{N = q} }+ 2 \eta ^{-1} \dP ( N \ne q)\\
&  \leq & \eta^{-2} q \dE | G_o (z) - g(z) |   +   2 \eta ^{-1} \DTV(P, \delta_q). 
\end{eqnarray*}
We find that for all $z \in \dC_+$ such that $\Im (z) > \sqrt q$, $G_o (z)$ converges in probability to $g(z)$ as $\DTV( P,\delta_q) \to 0$.  Since $z \mapsto G_o(z)$ is bounded by $ \Im(z) ^{-1}$ and analytic on $\dC_+$, from Montel's theorem, we may extend this convergence to all $z$ in $\dC_+$.  
Finally, by Lemma \ref{le:contextinc}, $\dP ( o \in S) = 1 - \pi_e$ goes to $1$ as  $\DTV( P,\delta_q) \to 0$. Hence, the same result holds for $G^s_o $. 
\end{proof}

\subsection{Tail bounds for the resolvent}
The next lemma can be found in  \cite[Proposition B.2]{aizenman2006}.
\begin{lemma}\label{le:unifintG}
For any $0< s < 1$ and $I = [a,b]$, there exists $C = C_s(a,b)$ such that for any probability measure $\mu \in \cP( \dR)$ and $\eta \geq 0$, 
$$
\int_I | g_{\mu} (\lambda+ i \eta) |^s d \lambda\leq C. 
$$
\end{lemma}

We fix a closed interval $I  = [a,b] \subset [-2 \sqrt q, 2 \sqrt q]$ and let $E$ be a random variable uniformly sampled on $I$. It was observed by Aizenman, Sims and Starr \cite{aizenman2006} that Lemma \ref{le:unifintG} implies the tightness of the random variables
$
G^s_o (E + i \eta)$,  $G^s_o (E + i 0)$ and $V  (E + i \eta)$.
Observe that $G^s _o (E + i 0)$ is well-defined since $G_o(\lambda+ i0) = \lim_{\eta \downarrow 0} G_o(\lambda + i \eta)$ exists for almost all $\lambda \in I$ and is measurable as a function of $\lambda$.  We will use the following corollary of Lemma \ref{le:unifintG}.
\begin{corollary}\label{cor:unifintG} Let $0< s < 1$. There exist positive constants $C , \veps$ such that, if $W_1 ( P, \delta_q) \leq \veps$, then 
for any $\eta \geq 0$ and $t > 0$, 
$$
\dP \PAR{  | G^s_o (E + i \eta) | \geq t}  \leq C t^{-s} \quad \hbox{ and } \quad \dP \PAR{  | G^s_o (E + i \eta) |^{-1} \geq t}  \leq C t^{-s}.
$$
\end{corollary}
\begin{proof}
The first statement follows directly from Markov inequality and Lemma \ref{le:unifintG}. From the second statement, we may observe that \eqref{eq:schur2} gives that 
$$
- (G_o^s (z) ) ^{-1} - z =     \sum_{x=1}^{N'_s}  G^s_x (z) + V (z) = g_{\nu} (z)
$$
is the Cauchy-Stieltjes transform of the finite measure $\nu$ on $\dR$ whose total mass is equal to 
$
\nu( \dR) =   N'_s + N'_e.
$
It follows, by Corollary \ref{cor:contextinc}, that $\dE \nu ( \dR) \leq C$ if $\veps$ is small enough. However, by Lemma \ref{le:unifintG} and the linearity of $\mu \mapsto g_{\mu}$, we get 
$$
\int_I |(G_o^s (\lambda + i \eta ) ^{-1} + \lambda + i \eta   |^s d\lambda = \int_I | g_{\nu} (\lambda+ i \eta) |^s d \lambda \leq C \nu(\dR)^s . 
$$
Taking expectation, we find
$$
\dE   |G^s_o (E + i \eta)^{-1} + E + i \eta   |^s  \leq  C \dE  \nu(\dR)^s \leq C_0.
$$
To conclude the proof, it remains to use $|x|^s \leq |x + y|^s + |y|^s$ and Markov inequality. 
\end{proof}

\begin{lemma}\label{le:vtheta}
Let $0 < s < 1$. For any $\eta \geq 0$, $\dE | V (E + i\eta) |^s \to 0$ as $W_1 ( P, \delta_q) \to 0$. %Moreover, if $Q_{e}$ denotes the law of $z \mapsto V(z)$ conditioned on $N'_e \geq 1$, as $W_1 ( P, \delta_q) \to 0$, we have $\DTV ( Q_{e} , \delta_{h} ) \to 0$, where $h(z) = - 1 / z$.  
\end{lemma}

\begin{proof}
We use that $V(z ) = g_{\nu } (z)$ where $\nu$ is a finite random measure with mass $\nu (\dR) = N'_e$.  In particular, by Lemma \ref{le:unifintG}, and the linearity of $\mu \mapsto g_{\mu}$, we find, for some $C> 0$, 
$$
\dE | V (E + i \eta) |^s  = \dE  \frac {1}{\ell(I)} \int_I | g_{\nu } (\lambda + i \eta ) |^s d\lambda \leq C \dE (N'_e )^s \leq C \dE N'_e.
$$
Now, $\dE N_e' \leq \dE N_e / (1 - \pi_e) = \pi_e \dE N / ( 1 - \pi_e) \leq \pi_e ( q + W_1 ( P , \delta_q) )  / ( 1 - \pi_e) $. We conclude with Lemma \ref{le:contextinc}. 
\end{proof}

\subsection{Lyapunov exponent}

In the spirit of \cite{aizenman2006}, for $z \in \dC_+$, we may then define the Lyapunov exponent 
$$
L_P(z) = -  \dE \log | G_o^s (z) | - \frac 1 2 \dE \log N'_s.
$$

\begin{lemma}\label{le:LP1}
The function $L_P$ is a non-negative harmonic function on $\dC_+$.  
\end{lemma}
\begin{proof}
First, $L_P$ is an harmonic function as it is the  real part of the harmonic function $- \dE \log G_o^s (z)  - \frac 1 2 \dE \log {N'_s}$.  Also, 
from \eqref{eq:schur2}, since $\Im (-1/z) = \Im (z) / |z|^2$, we have 
\begin{equation}\label{eq:Immod2}
\dE \log \Im ( G_o^s (z) ) = \dE \log |G_o^s (z) |^2 + \log \Im \PAR{z +  \sum_{x =1}^{{N'_s}} G_x^s (z) + V  (z) }.
\end{equation}
Now, $\Im(z) \geq 0$,  $\Im V  (z) \geq 0$ and $N'_s \geq 1$. Hence, from Jensen inequality,
\begin{eqnarray*}
\log \Im \PAR{z + \sum_{x =1}^{{N'_s}} G_x^s (z) + V (z) }&  \geq & \log \PAR{   \frac 1 {{N'_s}} \sum_{x =1}^{{N'_s}} \Im (G_x^s (z) )  } + \log {N'_s} \\
& \geq & \frac 1 {{N'_s}} \sum_{x = 1}^{{N'_s}} \log \Im (G_x ^s (z) )+ \log {N'_s}.
\end{eqnarray*}
We now take the expectation and use that $G_x^s$ are independent copies of $G_o^s$, independent of ${N'_s} \geq 1$. We obtain that 
$$
\dE \log \Im ( G_o^s (z) ) \geq 2 \dE \log |G_o^s (z) | + \dE \log \Im ( G_o^s (z) ) + \dE \log {N'_s}.
$$
Hence, $L_P (z) \geq 0$.  \end{proof}
From Corollary \ref{cor:unifintG}, if $W_1 (P, \delta_q)$ is small enough, for any $\eta \geq 0$, we may define  
$$
\gamma_{P} (\eta) =  \dE L_P ( E  + i \eta) =  - \dE \log | G_o^s (E + i \eta) | - \frac 1 2 \dE \log {N'_s},
$$
where, as above, $E$ is uniform on $I =[a,b]$. If $g \in \dC_+$ is defined by \eqref{eq:defgsc}, it is easy to check that we have for any $\lambda \in [-2,2]$, $|g( \lambda + i 0) |^2  = 1/ q$. Hence, for any $\lambda \in [-2,2]$, 
\begin{equation}\label{eq:gammadeltaq}
L_{\delta_q} (\lambda + i0) = 0 \quad \hbox{ and } \quad \gamma_{\delta_q} ( 0) = 0.
\end{equation}

The next key statement, first proved in a similar context in \cite{aizenman2006}, asserts that the averaged Lyapunov exponent is a continuous function of $(P,\eta)$. 
\begin{proposition}\label{prop:lyapunov}
Equip $\cP (\dZ_+)$ with the $W_1$-distance. Then, the function $(P, \eta) \mapsto \gamma_{P} (\eta)$ is continuous on $\delta_q \times [0,1]$, that is, for any $0 \leq \eta_0 \leq 1$, 
$$
\lim_{ W_1 (P, \delta_q) \to 0 , \eta \to \eta_0} \gamma_{P}(\eta)=  \gamma_{\delta_q} (\eta_0).
$$
In particular, for any $\veps > 0$, as $W_1 (P, \delta_q) \to 0$, 
$$
\ell \PAR{ \lambda \in I : L_P ( \lambda + i0 ) \geq \veps } \to  0.
$$
\end{proposition}

\begin{proof}
The second statement is a direct consequence of the first statement and \eqref{eq:gammadeltaq}. It is straightforward to adapt the proof of \cite[Theorem 3.1]{aizenman2006} (see there for a more detailed argument).  We first bound $ L_P(z)$ and write 
$$
| \log | G^s_o (z) |  | = \log_+  |G^s _o ( z) | +  \log_+  |G^s _o ( z) |^{-1},
$$ 
where $\log_+ (x) = \log(x \vee 1)$. For the first term, we have the bound $\log_+ |G^s _o ( z) | \leq \log_+  (\Im(z)^{-1} )$. For the second term, from  \eqref{eq:schur2}, 
$$
|G^s _o ( z) |^{-1} \leq \Im (z) + ({N'_s} + N_e ) \Im(z)^{-1}. 
$$ 
In particular, if $z = \lambda + i \eta$, with $\eta \geq 1$, 
$$
\dE | \log | G^s_o (z) |  |\leq 2 \dE \log  \PAR{ \eta +  \eta^{-1} ( {N'_s} + N_e) } \leq   2 \log  \PAR{ \eta +  \eta^{-1} \dE ( {N'_s} + N_e)  }.
$$
By Corollary \ref{cor:contextinc}, it follows that $L_{P} (\lambda + i \eta)/ \eta \to 0$ as $\eta \to \infty$ uniformly for all $\lambda \in \dR$ and $P$ such that  $W_1 (P, \delta_q) \leq \veps$ small enough.  We now fix $\eta \geq 0$. Since $z \mapsto L_P(z+i\eta)$ is a non-negative harmonic function on $\dC_+$, from Nevanlinna's representation theorem, it implies that 
$$
L_P (z + i \eta) = \Im  \int \frac{d \nu_{P,\eta}(\lambda)}{\lambda - z} = \Im \PAR{  g_{\nu_{P,\eta} }(z) },
$$
where $\nu_{P,\eta}$ is a Borel measure on $\dR$ such that $  \int \frac{d \nu_{P,\eta}(\lambda)}{1 + \lambda^2} < \infty$ (see Duren \cite{MR0268655}). Since $z \mapsto L_P(z + i \eta)$ has a definite sign, it has locally integrable boundary value \cite[Theorem 1.1]{MR0268655}. From the inversion formula of Cauchy-Stieltjes transform, we deduce that $ \nu_{P, \eta}$ is absolutely continuous with density at $\lambda$ given by $ L_P (\lambda+ i \eta) /\pi $ and  
$$
\gamma_{P} (\eta) =  \frac{1}{\ell(I)} \int_I L_P (\lambda+ i \eta) d\lambda =  \frac{1}{\pi \ell(I)} \int_I d \nu_{P, \eta} (\lambda) .
$$
Now, the claimed continuity of $(P, \eta) \mapsto \gamma_{P}(\eta)$ is a consequence of the vague continuity of $(P, \eta) \mapsto \nu_{P,\eta}$  on $\delta_q\times [0,1]$. This is in turn a consequence of the continuity for any $z \in \dC_+$ of $ (P, \eta) \mapsto \Im \PAR{  g_{\nu_{P,\eta} }(z) }$ on $\delta_q \times [0,1]$ (since the imaginary part of the resolvent characterizes the measure). Now, we recall that $\Im \PAR{  g_{\nu_{P,\eta} }(z) } = L_P ( z + i \eta)$,  hence this last continuity follows from (i) $\dE \log N'_s \to \log q$  as $W_1 ( P, \delta_q) \to 0$ and (ii) Lemma \ref{le:ASS} which implies that for all $z \in \dC_+$, $G_o^s (z)$ converges weakly to $g(z)$  when $\DTV(P,\delta_q) \to 0$.
\end{proof}

\subsection{Convergence of the resolvent on the real axis}
In this subsection, we will prove the following theorem.

\begin{theorem}\label{th:ASS}
Let $E$ be uniform on $I = [a,b]$, as $W_1 ( P , \delta_q) \to 0$, $(E,G_o (E + i0))$ converges weakly to $(E,g (E))$. Consequently, for any $\veps >0$, as $W_1 ( P , \delta_q) \to 0$,
$$
\int_I  \dP ( | g (\lambda) - G_o ( \lambda + i 0 ) | \geq \veps) d\lambda \to 0. 
$$
\end{theorem}
The main ideas of proof for the above result are again borrowed from \cite{aizenman2006}.  We will use a notion of discrepancy of a non-negative random variable $X$,  
$$
\kappa (X) =   \dE  \ABS{\frac{  X  -X' }{  X  + X' }}  
$$
with the convention that $0/0 = 0$ and $X'$ is an independent copy of the non-negative random variable $X$. It is easy to check that $\kappa (X) = 0$ is equivalent to $X$ a.s. constant. The next lemma summarizes some properties of $\kappa$. 
\begin{lemma}\label{le:discrep} Let $X,Y$ be non-negative random variables and $\lambda >0$. We have $\kappa( X + Y) \leq \kappa(X) + \kappa(Y)$, $\kappa (\lambda X) = \kappa (X)$, $\kappa(1/X) = \kappa(X)$ and, 
$$\kappa (X Y) \leq 6 \kappa(X) +   6 \kappa(Y).$$ 
Also, if $X = \sum_{i=1} ^N Y_i$ with $Y_i$ independent and independent of $N$, we have for any integer $q$, 
$$
\kappa \PAR{ \sum_{i=1} ^N Y_i }\leq \dP ( N \ne q) ^2  + \sum_{i=1} ^q \kappa ( Y_i).
$$
Finally, if $X_n$ converges weakly to $X$ then 
$$
\kappa(X) \leq \liminf_n \kappa (X_n).
$$
\end{lemma}
\begin{proof}
Only the last three statement deserves a proof. We write
\begin{eqnarray*}
\ABS{\frac{  XY  -X' Y'}{  X Y  + X'Y' }} &\leq &\ABS{\frac{ ( X   -X') Y}{  X Y  + X'Y' }} + \ABS{\frac{ ( Y   -Y') X' }{  X Y  + X'Y' }}\\
& \leq & ( T \vee 1) \ABS{\frac{ X   -X'}{  X  + X' }}  + (S \vee 1)  \ABS{\frac{  Y   -Y'}{  Y  + Y' }},
\end{eqnarray*}
with  $T= Y/Y'$ and  $S= X'/X$.  Now, if $k = (X' - X )/(X+X')$ and $l = (Y - Y' )/(Y+Y')$. We get, 
$$
T = \frac{1+ l}{1- l} \leq \frac{2}{1 -l} \quad \hbox{ and } \quad S = \frac{1+ k}{1- k}\leq \frac{2}{1 -k}. 
$$
Hence, from Markov inequality, if $t >2$,
$$
\dP ( T > t) \leq \dP ( l > 1 - 2 / t) \leq \kappa (Y)/(1 - 2/t),
$$
and similarly for $S$. We deduce that for $s,t >2$, 
$$
\kappa (X Y) \leq t \kappa(X) + s \kappa(Y) + \kappa (Y)/(1 - 2/t) + \kappa (X)/(1 - 2/s).
$$
We finally choose $s=t=3$ and get the required bound on $\kappa( XY)$.

If $X = \sum_{i=1} ^N Y_i$, we use that 
$$ 
\ABS{\frac{ \sum_{i=1} ^N Y_i -  \sum_{i=1} ^{N'} Y'_i}{  \sum_{i=1} ^N Y_i +  \sum_{i=1} ^{N'} Y'_i} } \leq \IND ( (N, N')  \ne ( q,q) ) + \ABS{\frac{ \sum_{i=1} ^q Y_i -  \sum_{i=1} ^{q} Y'_i}{  \sum_{i=1} ^q Y_i +  \sum_{i=1} ^{q} Y'_i} }.
$$

Finally, the statement about $\kappa(X_n)$ is a direct consequence of Fatou's lemma. \end{proof}
\begin{lemma}\label{le:kappa}
There exists $\veps >0$ such that if $W_1 (P, \delta_q) \leq \veps$ then, for any $z \in \dC_+$, 
$$
\kappa ( \Im  G^s_o ( z ) ) \leq \sqrt { 2 L_P (z)},
$$
and 
$$
\kappa ( | G^s_o ( z ) | ^2 ) \leq 6 \DTV ( P, \delta_q )^2 + 6 \sqrt 2 (q + 1)   \sqrt { L_P (z)}  + 6 \dP ( N'_e \geq 1).  
$$
\end{lemma}

\begin{proof}
We use the following second order refinement of Jensen inequality, for any integer $n \geq 2$ and positive $x_i$, 
$$
 \log \PAR{ \frac 1 n \sum_{i=1} ^n x_i } \geq \frac 1 n \sum_{i=1} ^n \log x_i + \frac{1}{2 n (n-1)} \sum_{i \ne j} \PAR{\frac{ x_i - x_j }{x_i + x_j}}^2,
$$
(proved in \cite[Lemma 4.1]{aizenman2006}). Let $z \in \dC_+$. From \eqref{eq:Immod2}, we find
 \begin{eqnarray*}
  \log \Im (G^s_o (z))& \geq &  \log | G^s_o (z) |^2 +  \frac 1 {{N'_s}} \sum_{x=1}^{{N'_s}}  \log \Im (G^s_x (z)) \\
& & \quad +  \frac{1}{2 {N'_s} ({N'_s}-1)} \sum_{x \ne y} \PAR{\frac{  \Im (G^s_x (z))  -  \Im (G^s_y (z)) }{ \Im (G^s_x (z))  +  \Im (G^s_y (z)) }}^2 + \log {N'_s},
\end{eqnarray*}
Since ${N'_s}$ is independent of $G^s_x$, taking expectation, we get that
 \begin{eqnarray*}
\dE  \frac{1}{ {N'_s}({N'_s}-1)} \sum_{x \ne y} \PAR{\frac{  \Im (G^s_x (z))  -  \Im (G^s_y (z)) }{ \Im (G^s_x (z))  +  \Im (G^s_y (z)) }}^2 \leq L_{P}(z).
\end{eqnarray*}
By Lemma \ref{le:contextinc}, if $ W_1 ( P, \delta_q) \leq \veps$ is small enough, then $\dP( {N'_s} \geq 2) \geq 1/2$, we deduce that
$$
\kappa ^2 ( \Im G^s _o(z))  \leq \dE \PAR{\frac{  \Im (G^s_1(z))  -  \Im (G^s_2 (z)) }{ \Im (G^s_1 (z))  +  \Im (G^s_2 (z)) }}^2 \leq 2 L_{P}(z),
$$
where $(G^s_i(z)) $, $i=1,2$ are independent copies of $G^s_o$. 

Similarly, from \eqref{eq:schur2}
$$
\Im (G^s _o(z) ) = |G_o ^s(z)|^2 \Im \PAR{ z +   \sum_{x=1}^{{N'_s}}  \Im (G^s_x (z)) + \Im V (z) }.
$$
We obtain from Lemma \ref{le:discrep} that
$$
\kappa ( |G_o ^s(z)|^2 )  \leq 6 \kappa  (\Im G^s _o (z)) + 6 \kappa \PAR{\sum_{x=1}^{{N'_s}}  \Im (G^s_x (z))}  + 6 \kappa( \Im V  (z)). 
$$
We find from another use of Lemma \ref{le:discrep},
$$
\kappa ( |G_o ^s(z)|^2 ) \leq 6 \DTV ( P, \delta_q )^2 + 6 (q + 1) \kappa  (\Im G^s _o (z ))  + 6 \kappa( \Im V  (z)).  
$$
Finally, $\kappa ( \Im (V(z) ) \leq  \dP ( V(z) \ne 0 ) \leq  \dP ( N'_e \geq 1 ).$ 
 It concludes the proof.
\end{proof}

We are now ready to prove Theorem \ref{th:ASS}.
\begin{proof}[Proof of Theorem \ref{th:ASS}]
First, by Lemma \ref{le:contextinc}, $\dP ( o \in S) \to 1$ as $\DTV (P, \delta_q) \to 0$. It thus sufficient to prove the result with the conditioned variable $G^s_o$ instead of $G_o$. The proof is then  essentially contained in \cite[Section 5]{aizenman2006}. Let us briefly sketch their argument.  From Corollary \ref{cor:unifintG}, the pair of random variables $(E,G^s _o (E))$ is tight as $W_1 ( P, \delta_q) \to 0$. Let us consider an accumulation point $(E,Z)$. From Corollary \ref{cor:unifintG}, $1/Z$ is a proper random variable on $\dC$. We denote by $P_\lambda$ the conditional distribution of $Z$ given $E = \lambda$ (it is defined for almost all $\lambda \in I$ from Fubini's Theorem). From the continuous mapping theorem, Lemma \ref{le:vtheta} and \eqref{eq:schur2},
\begin{equation*} 
- Z^{-1} \stackrel{d}{=}  E +  \sum_{i=1}^q Z_i,  
\end{equation*}
where, given $E$, $Z_i$ are independent copies of $Z$. Notably, if $Z(\lambda)$, $Z_i (\lambda)$ are independent with distribution $P_\lambda$, for almost $\lambda \in I$, 
\begin{equation}\label{eq:SZE}
- Z^{-1} (\lambda) \stackrel{d}{=}  \lambda +  \sum_{i=1}^q Z_i (\lambda).
\end{equation}
However, from Lemma \ref{le:kappa}, Proposition \ref{prop:lyapunov} and Fubini's Theorem,
$$
\int_I \dE \kappa ( \Im (Z(\lambda) )   d\lambda = \int_I \dE \kappa (  |Z(\lambda)|^2 )   d\lambda = 0.
$$
We deduce that for almost all $\lambda \in I$, $\Im (Z(\lambda)$ and $|Z(\lambda)|$ are supported on a single point. In particular, $P_\lambda$ is supported on at most $2$ points. Assume that $Z(\lambda)$ can take two distinct values with positive probability, say ($z_1, z_2$) . From \eqref{eq:SZE}, since $q \geq 2$, $- Z^{-1} (\lambda)$ could take at least three different  values  with positive probability :  $(\lambda + q z_1, \lambda + q z_2 , \lambda + ( q -1) z_1 + z_2)$. It contradicts the fact that the support of $P_\lambda$ has at most two points.

Finally, if, for almost all $\lambda \in I$,  $P_\lambda = \delta_{z(\lambda)}$ then, from \eqref{eq:SZE}, $- z(\lambda) ^{-1} = \lambda + q z(\lambda)$. Since $\Im ( z(\lambda) )\geq 0$ it implies that $z(\lambda) = g(\lambda)$.  \end{proof}

\subsection{Proof of Theorem \ref{th:GW}, Corollary \ref{cor:GW} and Theorem \ref{th:UGW}}

We start by recalling the probabilistic version of Scheff\'e's Lemma. 
\begin{lemma}[Scheff\'e's Lemma]\label{le:scheffe}
Let $X_n$ be a sequence of non-negative random variables converging in probability to $X \in L^1 (\dP)$.  Then $\dE X_n \to \dE X$ implies $\dE |X_n - X| \to 0$. 
\end{lemma}
\begin{proof}
By dominated convergence, $\dE (X_n \wedge X ) \to \dE X$. However, $|X_n - X| =  X + X_n - 2 (X_n \wedge X)$. 
\end{proof}

\begin{proof}[Proof of Theorem \ref{th:GW}]
Let $f$ be the density of the absolutely continuous part of the random measure $\mu_T^{e_o}$.
 We have for almost all $\lambda$, $f(\lambda) =  \Im  G_o  ( \lambda + i 0 )  / \pi $. From Theorem \ref{th:ASS} applied to $I = [-2 \sqrt q , 2 \sqrt q]$, for any $\veps > 0$, if $W_{1} ( P , \delta_q)$ is small enough, 
$$
\dE \int  f (\lambda) d \lambda \geq  \int_{-2\sqrt q } ^{2 \sqrt q} f_q ( \lambda) d\lambda - \veps = 1  -\veps. 
$$
It follows that 
$$
\lim_{ P \Lto \delta_q} \dE \int  f (\lambda) d \lambda = 1. 
$$
It remains to use Theorem \ref{th:ASS} with Scheff\'e's Lemma \ref{le:scheffe} for the probability measure $\dP \otimes U  $ where $U$ is the uniform probability  measure on $ [- 2 \sqrt q , 2 \sqrt q]$.  \end{proof}

\begin{proof}[Proof of Corollary \ref{cor:GW}]
By definition,  $\bar f(\lambda) \geq \dE f(\lambda)$. Hence from Theorem \ref{th:GW}, for any Borel $I$, 
$$
\lim_{P \Lto \delta_q} \int_I \bar f(\lambda) d \lambda \geq \int_I  f_q(\lambda) d \lambda.
$$  
Applied to $I =  [- 2 \sqrt q, 2 \sqrt q]$, we deduce that the above inequality is an equality.  That is, for any Borel $I$, 
$$
\lim_{P \Lto \delta_q} \int_I \bar f(\lambda) d \lambda =  \int_I  f_q(\lambda) d \lambda.
$$  
We conclude with a new use of  Scheff\'e's Lemma \ref{le:scheffe}. 
\end{proof}

\begin{proof}[Proof of Theorem \ref{th:UGW}]
We start by observing that $W_{2} ( P , \delta_{q+1}) \to 0$ implies that $W_1 ( \widehat P, \delta_q) \to 0$.  Also, from Schur's formula \eqref{eq:schur}
$$
G_o (z) \stackrel{d}{=} - \PAR{ z + \sum_{i=1}^N \hat G_i (z)}^{-1},
$$
where $N$ has distribution $P$, independent of $(\hat G_i (z))_{i \geq 1}$, independent copies of $\hat G_o(z)$, the resolvent at the root of the adjacency operator of a $\GW(\hat P)$ tree. 

Now, the Cauchy-Stieltjes transform of $\check f_q$ is $\check g$ which satisfies the identity,
$$
\check g (z) = - ( z + (q+1) g(z) ) ^{-1},
$$
where $g(z)$ denotes the Cauchy-Stieltjes transform of $f_q$ (see \cite[Eqn (5)]{MR2724665}). It remains to apply Theorem \ref{th:ASS} and the continuous mapping theorem. We find that, if $E$ is uniform on $I = [a,b]$, then
$
(E,G_o ( E + i 0) )
$ converges weakly to $(E,\check g (E))$. We may then repeat the argument of Theorem \ref{th:GW} and Corollary \ref{cor:GW}.
\end{proof}

%Observe also that $f_\eta (\lambda)= \frac 1 \pi \Im  G^s _o  ( \lambda + i \eta )$ is the density at $\lambda$ of $P_\eta * \mu^{e_o}_A$ where $P_\eta$ the Cauchy distribution with density $\eta dx / ( \pi  (\eta^2 + x^2)) $. It thus follows from Scheff\'e's Lemma that 
%$$
%\lim_{\eta \to 0} \int f_\lambda 
%$$

\subsection{Other approaches to the existence of continuous spectrum}

In the above argument, we have followed the strategy of \cite{aizenman2006} to prove the existence of continuous spectrum for random Schr\"odinger operators on infinite trees at small disorder. Other approaches of this result have been proposed, they all start from the analog of the recursive distribution equation coming from Schur's formula \eqref{eq:schur}. The original proof of Klein \cite{klein} relies on an application of the implicit function theorem, see also \cite{MR2864550,MR2842363}, a more geometric study of the fixed point equation was initiated by Froese, Hasler and Spitzer \cite{MR2863854} and further developed in \cite{MR3051700,MR2274470,MR2913621,MR3070682,MR2994759}.
 
A common point of all these methods is that they require more moment conditions on the disorder than the one obtained in Lemma \ref{le:vtheta}. In fact, if we restrict our attention to a specific Borel subsets we can improve drastically on Lemma \ref{le:vtheta}. This is the content of the next statement. 

\begin{proposition}\label{le:vtheta2}
Let $\veps > 0$ and $p \geq 1$. There exists a Borel set $K \subset \dR $ (depending on $\veps$ and $p$) such that $\ell ( K^c ) \leq \veps$ and, as $W_p ( P, \delta_q) \to 0$,
 $$\sup\BRA{\dE | V (\lambda + i \eta) |^p : \eta \geq 0 , \lambda \in K  } \to 0.$$ 
\end{proposition}

Before proving this proposition, let us simply mention that it could be used to give alternative proofs of Theorem \ref{th:GW} and Theorem \ref{th:UGW} following the approach of Keller et al. \cite{MR3070682,MR2994759}. The approach of Klein does not seem however to accommodate easily with vertices with only one offspring in the skeleton tree. We will however not pursue further in this direction here and restrict ourselves to the proof of Proposition \ref{le:vtheta2}.

We start with a standard lemma on the total progeny of subcritical Galton-Watson trees. 

\begin{lemma}[Total progeny of subcritical Galton-Watson tree]\label{le:totalprogeny}
Let $Q$ be a probability measure on non-negative integers whose moment generating function $\psi$ satisfies $\psi(\rho) < \rho$ for some $\rho > 1$. Let $Z$ be the total number of vertices in a $\GW(Q)$ tree. We have for any $t \geq 1$,
$$
\dP ( Z \geq t) \leq \rho \PAR{ \frac{ \psi( \rho)}{\rho}}^t.
$$
\end{lemma}

\begin{proof}
The proof is extracted from \cite[Theorem 2.3.1]{MR2656427}. Let $Y_i, i \geq 1$, be iid copies with distribution $Q$ and $X_i = Y_i -1$. Consider the random walk,  $S_0 = 1$ and for $t \geq 1$, $S_t = S_0 + \sum_{i=1} ^t X_i$.  The total progeny $Z$ has the same distribution that the hitting time $\tau = \inf \{ t \geq 1 : S_t = 0\}$ (see for example \cite[Section 2.1]{MR2656427}). Set $\theta = \log \rho$, $f(\theta) = \dE e^{\theta X} = \psi ( e^{\theta}) e^{-\theta} = \psi (\rho) / \rho < 1$. By construction $M_t = e^{\theta S_t} / f(\theta) ^t$ is non-negative martingale with mean $M_0 = e^{\theta} = \rho$ with respect to the filtration $\cF_t = \sigma (S_0, \cdots, S_{t})$. From Doob's optional stopping time theorem, we have 
$$
\dE[ M_\tau]= \dE[f (\theta) ^{- \tau} ]= \rho.   
$$
Then, since $0 < f < 1$, from Markov inequality, 
$$ \PP ( \tau \geq  t ) = \PP( f(\theta) ^{-\tau} \geq  f(\theta)^{-t}) \leq  \rho  f(\theta)^{t}.$$
Since $\dP ( Z \geq t) = \PP ( \tau \geq  t )$ it concludes the proof. 
\end{proof}

\begin{proof}[Proof of Proposition \ref{le:vtheta2}]
We first observe that  for any $0 <\delta < 1$ and all $0 \leq y \leq y_\delta = (\delta/2)^{1/(q-1)}$, we have $y^q \leq  (\delta / 2)  y$. The parameter $\delta >0$ will be fixed later on in the proof. We also fix some $0 < y < y_\delta$.

As usual, let $\varphi$ be the moment generating function of $P$.  Now, for $x \in [0,1]$, since for any $a,b \geq 0$, $|x^a - x^b| \leq |a - b|$, we find that  $\ABS{ \varphi (x)  -  x^{q} } \leq \dE | N - q| = W_1(P, \delta_q)$. From what precedes for all $P$ such that $W_1 (P, \delta_q) \leq \delta y /2$, we find that $\varphi ( y ) \leq \delta y$.

From \cite[Section 12, Theorem 3]{MR0373040}, $G^e_o(z)$ is the resolvent at the root of the adjacency matrix of a random tree $T$, a subcritical $\GW(Q)$ tree where $Q$ has moment generating function $\varphi_e$ is given by \eqref{eq:varphie}.  From  what precedes, we deduce that $\varphi_e ( \rho) \leq \delta \rho$, with $\rho = y / \pi_e$. By Lemma \ref{le:contextinc}, if $W_1 (P, \delta_q)$ is small enough then $\rho > 1$. Hence, from Lemma \ref{le:totalprogeny}, if $|T|$ is the total number of vertices in $T$, we deduce that for all $k\geq 1$,
\begin{equation}\label{eq:tailboundT}
\dP ( |T| \geq k ) \leq \frac{y}{\pi_e} \delta^k. 
\end{equation}
as soon as $W_1 (P, \delta_q)$ is small enough.

Now, for integer $k \geq 1$, let $\Lambda_k$ be the set of real numbers $\lambda$ such that there exists a tree with $k$ vertices and $\lambda$ is an eigenvalue of this tree. Obviously, $|\Lambda_k|$ is bounded by $k$ times the number of unlabeled trees with $k$ vertices. In particular, for some $c > 1$, 
\begin{equation}\label{eq:ubLk}
| \Lambda_k | \leq c^k,
\end{equation} see Flajolet and Sedgewick \cite[Section VII.5]{MR2483235}. We define $B_{k,\veps} = \{ x \in \dR : \exists \lambda  \in \Lambda_k,  | \lambda - x | \leq \veps 2^{-k} / | \Lambda_k| \}$ and $K = \dR \backslash \cup_{k \geq 1} B_{k,\veps}$. By construction, 
$$
\ell (  K^c ) \leq \sum_{k  \geq 1} | \Lambda_k | \frac{ \veps 2^{-k} }{ | \Lambda_k| }= \veps. 
$$ 

Also, we have for any probability measure $\nu$, $|g_\nu ( z) | \leq 1 / d ( z , \supp ( \nu) ) $. Hence, we find from \eqref{eq:tailboundT}-\eqref{eq:ubLk}, for any $\lambda \in K$, 
\begin{eqnarray*}
\dE | G^e _o (\lambda + i \eta) |^p &= & \sum_{k=1}^\infty \dP ( |T | = k ) \dE \left[ | G^e _o (\lambda + i \eta) |^p | |T| = k \right]\\
& \leq & \sum_{k=1}^\infty \frac{y}{\pi_e} \delta^k  \PAR{  2^{k}  | \Lambda_k| / \veps }^p \\
& \leq & \frac{y}{\pi_e\veps^p } \sum_{k=1}^\infty     \PAR{  \delta (2c)^{p}     }^k 
\end{eqnarray*}
Hence, if $\delta$ was chosen such that $\delta ( 2 c)^p \leq 1/2$, we obtain,
$$
\dE | G^e _o (\lambda + i \eta) |^p \leq \frac{y}{\pi_e \veps^p }.
$$

Finally, from the definition of $V(z)$ in \eqref{eq:defV}, we have, using H\"older inequality,
\begin{eqnarray*}
\dE \ABS{ V ( \lambda + i \eta) }^p & = &  \dE \ABS{ \sum_{x=1} ^{N'_e} G_x^e ( \lambda + i \eta) }^p  \\
& \leq & \dE \BRA{ (N'_e)^{p-1} \sum_{x=1} ^{N'_e}   | G^e _x (\lambda + i \eta) |^p } \\
& =  &  \dE \BRA{ (N'_e)^{p}} \dE  | G^e _o (\lambda + i \eta) |^p.
\end{eqnarray*}

Now, $\dE (N'_e )^p \leq \dE N_e ^p / ( 1 - \pi_e)$ and, as already pointed, $N_e \stackrel{d}{=} \sum_{i=1}^N \veps_i$ where $N$ has distribution $P$, independent of $(\veps_i)_{i \geq 1}$ an i.i.d. sequence of Bernoulli variables with  $\dP ( \veps_i  = 1) = \pi_e = 1 - \dP ( \veps_i = 0)$.   In particular, H\"older inequality implies that 
$$
\dE N_e ^p \leq \dE N ^p \dE \veps ^p = \dE N ^p \pi_e. 
$$

We have thus proved that 
$$
\dE \ABS{ V ( \lambda + i \eta) }^p  \leq  \frac{y \dE N^p}{( 1- \pi_e)  \veps^p } .
$$
Since $y$ can be taken arbitrarily small as $W_1 (P, \delta_q) \to 0$, the conclusion follows. \end{proof}

\section{Deterministic resolvent bounds}
\label{sec:DRB}

In this section, we state some general relations involving resolvent matrix and delocalization.

\subsection{Convergence and matching moments}

The objective of this subsection is to compare the Stieltjes transforms of two measures whose first moments coincide.

\begin{proposition}\label{prop:jackson}
Let $\mu_1 , \mu_2$ be two real probability measures such that for all integers $1 \leq k \leq n$, 
$$
\int \lambda^k d \mu_1(\lambda) = \int \lambda^k d \mu_2(\lambda).
$$
Let $\zeta = e^2 \pi$.  If $\mu_1$ and $\mu_2$ have support in $[-b,b]$ then for all $z \in \dC_+$ with  $\Im(z) \geq \zeta b \lceil \log n \rceil / n$,
$$
\ABS{ g_{\mu_1} (z) - g_{\mu_2} (z) } \leq \frac{2}{\zeta n b }.
$$
\end{proposition}

\begin{proof}
We set
$$ g_z(\lambda) = \frac 1 {\lambda -z}.$$
For integer $k \geq 0$, we have
$$
\| \partial^{(k)} g_z \|_{\infty}  =  k! \, \eta^{-k-1}. 
$$
From Jackson's theorem \cite[Chap. 7, \S 8]{MR1261635}, there exists a polynomial $p_z$ of degree $n$ such that for any $\lambda \in [-b,b]$ and $k \leq n$,  
\begin{equation}\label{eq:jackson0}
\ABS{ g_z ( \lambda) - p_z(\lambda) } \leq \PAR{ \frac \pi 2  b }^k \frac{ (n - k +1)!}{(n +1) ! }  \| \partial^{(k)} g_z \|_{\infty}.
\end{equation}
We take $k = \lceil \log n \rceil$ and $\eta \geq \zeta b \lceil \log n \rceil / n$. Using, $k ! \leq k^k$, $\log n / n \leq e^{-1}$, we get, 
$$
\ABS{ g_z ( \lambda) - p_z(\lambda) } \leq \frac { 1 } { \eta }  \PAR{ \frac{ \pi  b  k }{ 2 \eta (n+2 -k)}}^{k} \leq \frac { 1 } { \eta }  \PAR{ \frac{1}{  2  e^2 } \frac{1}{ 1 - e^{-1}}  }^{k} \leq \frac 1 {\eta n^2} \leq \frac 1 { \zeta  b  n}.  $$
The conclusion follows.
\end{proof}

As an immediate corollary, we have the following statement. 
\begin{corollary}\label{cor:jackson}
For $i = 1, 2$,  let $(G_i,o)$ be a rooted graph and denote by $A_i$ their adjacency operators. Assume further that the rooted subgraphs $(G_1,o)_h  $ and $(G_2,o)_h$ are isomorphic. If, for $i = 1, 2$, $\|A_i \| \leq b$ then for all $z \in \dC_+$ with $\Im (z) \geq \zeta  b \lceil \log 2h \rceil / (2h)$, 
$$
\ABS{ \langle e_o , (A_1 - z)^{-1} e_o \rangle -\langle e_o , (A_1 - z)^{-1} e_o \rangle } \leq  \frac { 1}{ \zeta b h}.
$$
\end{corollary}
\begin{proof}
By assumption and \eqref{eq:defmuGex}, we can apply Proposition \ref{prop:jackson} to $n = 2 h$. 
\end{proof}

\subsection{Regularity and resolvent}

In this paragraph, for any interval $I$, we state a weak bound of $\mu(I)$ in terms of $g_{\mu}$. Much stronger statements have appeared in the literature, see for example \cite{ESY10,TV11,MR3129806}. There are however not really adapted to quantum percolation due to the typical presence of  a dense atomic part. The next statement gives a weak regularity result.

\begin{lemma}\label{le:weakdeconv}
Let $\mu$ be a probability measure on $\dR$ such that for some $\lambda \in \dR$ and $a,b,\eta >0$, $$\Im ( g_{\mu} ( \lambda + i \eta ) ) \geq a \quad \hbox{ and for all $y \geq \eta$, } \quad \Im ( g_{\mu} ( \lambda + i y ) ) \leq b.$$ 
Then, if $I = [\lambda - s /2 , \lambda + s  /2]$ or $I =  (\lambda - s /2 , \lambda + s  /2)$ with $\ell(I) = s  \geq 2 \eta$, we have
$$
\frac { a }{2 \rho }  \leq \frac{\mu (I) }{\ell(I)} \leq b,
$$
where the left-hand side inequality holds if $\rho = s / \eta \geq 8 b / a$. 
\end{lemma}
 
\begin{proof}
We have the bound,
\begin{eqnarray}\label{eq:lbst}
\Im g_{\mu} ( x + i y) = \int \frac{ y }{ (x - \lambda)^2 + y ^2} d\mu(\lambda) \geq \frac{\mu \PAR{[ x - y , x + y ] }}{2 y}.
\end{eqnarray}
Applied to $x = \lambda$ and $y = s/2$, it readily implies the upper bound of the lemma.

For the lower bound, let $I_0  = I= (x - t \eta , x + t \eta)$ and for $k \geq 1$, $I_k = [x - t 2^k \eta , x - t 2^k \eta] \backslash I_{k-1}$. We write 
$$
\Im g_{\mu} ( x + i \eta) \leq \frac{\mu (I)}{ \eta } + \sum_{k=1}^\infty \frac{\mu (I_k)}{ \eta ( 1+ 4^k t^2 ) }  \leq \frac{\mu (I)}{ \eta } + \sum_{k=1}^\infty \frac{\mu (I_k)}{ \eta 4^k t^2  }
$$
From \eqref{eq:lbst}, for $x = \lambda$ and $y = t 2^k \eta$, $ \mu (I_k) \leq  \mu ( [x - t 2^k\eta , x - t 2^k \eta] ) \leq 2^{k+1} t \eta b$, and
$$
 \frac{\mu (I)}{ \eta }  \geq  a  -  \sum_{k=1}^\infty \frac{  2 b }{  2 ^k t  } = a - \frac{2 b}{t}. 
$$
Setting $2 t \eta = s = \ell ( I)$, we deduce that 
$$
\rho  \frac{\mu (I)}{ \ell(I) } \geq  a - \frac{4 b}{\rho} \geq \frac a 2,
$$
if $\rho \geq 8 b /a$.
\end{proof}

Let $G$ be a finite graph with $n$ vertices, $A$ its adjacency matrix, and $\lambda_1, \cdots, \lambda_n$ its eigenvalues. For $I \subset \dR$, we denote by 
$$
\Lambda_I = \{ k : \lambda_k  \in I\}.
$$
By definition, we have $\mu_G (I) = n | \Lambda _I |$. In the sequel, $\psi_1, \cdots, \psi_n$ is a orthornormal basis of eigenvectors of $A$, $A \psi_k = \lambda_k \psi_k$.  Finally, the resolvent of $A$ is denoted by $
R(z) = ( A - z I) ^{-1}$.

\begin{corollary}\label{cor:psik}
Let $G$ be as above and let $o \in V(G)$. Assume that for some $\lambda \in \dR$ and $a,b,\eta >0$, $$\Im ( R_{oo} ( \lambda + i \eta ) )  \geq a \quad \hbox{ and for all $y \geq \eta$, } \quad \Im ( R_{oo} ( \lambda + i y ) ) \leq b.$$ 
Then, if $I = [\lambda - s /2 , \lambda + s  /2]$ or $ I = (\lambda - s /2 , \lambda + s  /2)$ with $\ell(I) = s  \geq 2 \eta$, we have
$$
 \frac { a s }{2 \rho }  \leq \sum_{k \in \Lambda_I}  | \psi_k (o) |^ 2 \leq b s,
$$
where the left-hand side inequality holds if $\rho = s / \eta \geq 8 b / a$. 
\end{corollary}

\begin{proof}
We recall that  $R_{oo} (z) = g_{\mu^{e_o}_G} (z)$ and $\mu^{e_o} _G  (I ) = \sum_{ k \in \Lambda_I}  | \psi_k (o) |^ 2$. It thus remains to apply Lemma \ref{le:weakdeconv}.
\end{proof}

The following elementary lemma will also be useful. 
\begin{lemma}\label{le:coaire}
Let $U$ be an open set, $t > 0$ and $\mu$ a probability measure on $\dR$, then 
$$
\mu( U)\leq 2 \int_U \Im g_{\mu} (   x + i t ) dx 
$$ 
\end{lemma}
\begin{proof}
From  \eqref{eq:lbst}, it is sufficient to prove that 
$$
\mu(U) \leq \frac{1}{ t} \int_U \mu ( [x - t , x + t ]) dx. 
$$
Assume first that $U = (a,b)$ is an open interval and that $\mu$ is absolutely continuous with density $f$. Then, we write
$$
\int_U \mu ( [x - t , x + t ]) dx =    \int_{a} ^b  \int_{x - t} ^{x +t} f (y) dy  dx = \int_{a-t}^{b+t} f (y) \int_{y-t}^{y+t}  \IND ( x \in (a ,b) ) dx dy.
$$
For any $y \in (a,b)$, $\int_{y-t}^{y+t}  \IND ( x \in (a ,b) ) dx \geq t$. We deduce that 
$$
\int_U \mu ( [x - t , x + t ]) dx \geq t \int_{a}^{b} f (y) dy = t \mu ( U). 
$$
Since, any open set in $\dR$ is a countable union of disjoint intervals. We obtain by linearity the claimed result when $\mu$ is absolutely continuous. In the general case, we consider a sequence of absolutely continuous probability measures $\mu_n$ which converge weakly to $\mu$. Since $U$ is open and $[x - t , x + t ]$ is closed, $\mu( U) \leq \liminf \mu_n (U)$ and   $\mu( [x - t , x + t ] ) \geq \limsup \mu_n ([x - t , x + t ])$. We may thus take the limit in the inequality for $\mu_n$. 
\end{proof}

\section{Rates of convergence in percolation graphs}

\label{sec:rate}
In this section, we prove Theorem \ref{th:rate}.

\subsection{Concentration Lemma}
We start by recalling a useful concentration lemma in the context of percolation. Recall that the total variation norm of $f:\dR \to\dR $ is
$$
\|f \| _\textsc{TV}:=\sup \sum_{k \in \dZ} | f(x_{k+1})-f(x_k) |, \, 
$$
where the supremum runs over all sequences $(x_k)_{k \in \dZ}$ such that
$x_{k+1} \geq x_k$ for any $k \in \dZ$. If $f = \IND_{(-\infty,s]}$ for
some real $s$ then $\|f \|_\textsc{TV}=1$, while if $f$ has a derivative in
$\mathrm{L}^1(\dR)$, we get
$
\|f \| _\textsc{TV}=\int |f'(t)|\,dt.
$ The following lemma is a consequence of \cite[Lemma C.2]{MR2837123}.

\begin{lemma}\label{le:concperc}
Let $p \in [0,1]$ and $H = \perc (G, p)$ where $G$ is a finite deterministic graph. Then, for any $f:\dR\to\dC$ such that
  $\| f \|_\textsc{TV}\leq1$ and every
  $t\geq0$,
  \[
  \dP \left( \left| \int\!f\,d\mu_H -\dE\int\!f\,d\mu_H \right|  \geq t \right) %
  \leq 2 \exp\left({-\frac{n t^2}{8}}\right).
  \]\end{lemma}

\subsection{Proof of Theorem \ref{th:rate}}

We have 
\begin{eqnarray*}
\dE \int \varphi d \mu_{\perc(G,p)}- \int \varphi d \mu  & = & \frac 1 n \sum_{v=1} ^n \PAR{ \dE \int \varphi d \mu^{e_v}_{\perc(G,p)} - \dE \int \varphi d \mu^{e_o}_{\perc(\Gamma,p)}}.
%& =  & \frac{(n - B(h))}{n}  \PAR{ \dE \int \varphi d \mu^{e_o}_A - \int \varphi d \mu^{e_o}_{B}}
\end{eqnarray*}
Now, if $v \in B(h)$, then $$\ABS{ \dE \int \varphi d \mu^{e_v}_{\perc(G,p)} - \dE \int \varphi d \mu^{e_o}_{\perc(\Gamma,p)}}   \leq 2 \| \varphi \|_\infty.$$
Otherwise, if $v \notin B(h)$, then $ \dE  \int x^k  d  \mu^{e_v}_{\perc(G,p)}   =  \dE  \int x^k   d  \mu^{e_o}_{\perc(\Gamma,p)}$ for all $k \leq 2h$. Hence, Jackson's Theorem (see \eqref{eq:jackson0}) implies that 
$$
\ABS{ \dE \int \varphi d \mu^{e_v}_{\perc(G,p)} - \dE \int \varphi d \mu^{e_o}_{\perc(\Gamma,p)}}  \leq 2  \PAR{ \frac \pi 2  d }^k \frac{ (2h - k +1)!}{(2h +1) ! }  \| \partial^{(k)} \varphi \|_{\infty}.$$

To conclude the proof of Theorem \ref{th:rate}, it remains to use Lemma \ref{le:concperc}.

\subsection{Proof of Corollary \ref{cor:rate}}

Let $\tilde \mu = \mu_{\perc(G,p)}$.  Using Corollary \ref{cor:jackson} in the proof of Theorem \ref{th:rate}, we find for $\Im (z) \geq \eta_1 = \zeta d \lceil \log 2h \rceil  /(2 h)$,
\begin{equation*}\label{eq:tlsoo}
\ABS{  g_{\tilde \mu} (z) - g_\mu (z) }  \leq t + \frac{4}{ \zeta d \log 2h } \frac{h  B_\Gamma(h)}{n}  + \frac{1}{\zeta d h}  \leq  t + \frac{ 5 \delta }{\zeta d} .
\end{equation*}
with probability at least $1 - 2 \exp ( - n t^2 / ( 8 ( \pi  / \eta_1  ) ^2  )$. For $t = 18  \delta /(\zeta d)$, we deduce that $\ABS{  g_{\tilde \mu} (z) - g_\mu (z) }  \leq 23  \delta / (\zeta d) $ with probability at least $1 - 2 \exp ( -  c n \delta^2 / h^2 )$ with $c = 18^2  / ( 32  \pi^2) \geq 1$.  It remains to use the numerical value of $\zeta \in (23,24)$. 

\section{Weak delocalization in percolation graphs}
\label{sec:deloc}

In this section, we prove Theorem \ref{th:percdeloc} and Theorem \ref{th:perclocallaw}. 

\subsection{Proof of Theorem \ref{th:perclocallaw}}

Let $S = [-d,d]$, $f$ be the density of the absolutely continuous part of $\mu $ and $h_{\eta} (\lambda) = \frac 1 \pi \Im g_{\mu} (z)$ with $z = \lambda + i \eta$.  Recall that a.e. $ f (\lambda) = \lim_{\eta \to 0}  h_\eta(\lambda)$. By assumption, 
$$
m =  \int_S  f (\lambda ) d\lambda   = \mu_{ac} (\dR),
$$
is positive. By monotone convergence, for any $0< \veps < m$, we can find a pair $(a,b)$ of positive numbers satisfying, 
$$
 \int_S    f (\lambda )   \IND ( a \leq   f (\lambda ) \leq b) d\lambda  \geq m  - \veps/2.  
$$
Also, if $a' = a/2$, $b' = 2b$, and $\eta_0 > 0$, we define the Borel set 
$$
K_0 = \{ \lambda  \in S  : h_\eta( \lambda) \in [a',b'] \hbox{ for all } \eta \in [0,\eta_0] \}.
$$
From Egorov's Theorem, if $\eta_0$ is small enough, for any $0 \leq \eta \leq \eta_0$, 
\begin{equation}\label{eq:hetal1}
\int_{K_0}   h_\eta( \lambda)  d\lambda  \geq m -  \veps.  
\end{equation}
Using the bounds,  for all $\lambda \in S$ and $\eta_0 \leq \eta \leq 1$, $\eta_0/( 4 d^2 + \eta_0^2) \leq h_\eta (\lambda) \leq 1/\eta_0$, we find that for some positive $a'',b''$,
$$
K_0 \subset K_1 =  \{ \lambda  \in S  :   h_\eta( \lambda) \in [a'',b''] \hbox{ for all } \eta \in [0,1] \}.
$$

We now set $\tilde \mu = \mu_{\perc (G,p)}$ and  $ \tilde h_{\eta} (\lambda) = \frac{1}{\pi} \Im g_{\tilde \mu} (z)$,  $z = \lambda + i \eta$. By Corollary \ref{cor:rate}, for any $\eta \geq \eta_1 =  20 d   \log ( 2 h )  / h$ and $\lambda \in \dR$, we have
$$
\ABS{ \tilde h_\eta (\lambda) - h_\eta (\lambda)} \leq  \delta / \pi
$$
with probability at least $1 - 2 \exp ( - n \delta^2 /  h^2)$. Recall that $| g_{\mu} (z) - g_{\mu} (z') | \leq |z - z'| / ( \Im (z) \wedge \Im (z') )^2 $.  Consider a $\delta \eta_1^2$-net of $(4 d  +2 ) /( \delta \eta_1^2 ) \leq h^2 / \delta$ points on the boundary of the rectangle $R = \{ z :  \eta_1 \leq \Im ( z) \leq 1 , \Re (z) \in S \}$. From the maximum principle and the union bound, we deduce easily that 
\begin{equation}\label{eq:supconv}
\sup_{\lambda + i \eta \in R} \ABS{ \tilde h_\eta (\lambda) - h_\eta (\lambda)} \leq \frac {2 \delta}{\pi} \leq \delta.
\end{equation}
with probability  at least $1 -  \delta^{-1} h^2  \exp ( - n \delta^2 / h^2)$. 

Up to modifying the constants $c_0,c_1$ in the statement of the theorem, we can assume without loss of generality that $\eta_1 \leq \eta_0$. Hence, on the event \eqref{eq:supconv}, for all $\lambda \in K_1$, $\eta_1 \leq \eta \leq 1$, $\tilde h_\eta( \lambda) \in [a'' - \delta,b'' + \delta]$. It remains to apply Lemma \ref{le:weakdeconv} to all $ \lambda \in K_1$ and set $K = \bar K_1$. We deduce the first statement of Theorem \ref{th:perclocallaw} by adjusting all constants.

Also, since $\int \tilde h_\eta (\lambda ) d\lambda = 1$, we find from \eqref{eq:hetal1} that, on the event \eqref{eq:supconv},
$$
\int_{K^c} \tilde h_{\eta} (\lambda) d\lambda \leq 1 - m +  \veps + \delta. 
$$
Applying Lemma \ref{le:coaire}, we obtain the second statement of Theorem \ref{th:perclocallaw}.

\subsection{Concentration lemma for local graph functionals}

To prove Theorem \ref{th:percdeloc}, we need a basic concentration lemma for local functions of the graph. To this end, we denote  by $\cG^*$ the set of finite rooted graphs, i.e. the set of pairs $(G,o)$ formed by a finite graph $G = (V,E)$ and a distinguished vertex $o \in V$. Recall that for integer $h \geq 1$, we denote by $(G,o)_h$ the subgraph of $G$ spanned by the vertices which are at distance at most $h$ from $o$. We shall say that a function $\tau$ from $\cG^*$  to $\dR$ is $h$-local, if  $\tau(G,o)$ is only function of $(G,o)_h$.

The next statement is a straightforward corollary of Azuma-Hoeffding's inequality. Recall that the parameter $M_h(G)$ was defined in \eqref{eq:NhMh}.
\begin{lemma}
\label{th:stabloc}
Let  $p \in [0,1]$ and $H = \perc (G, p)$ where $G$ is a deterministic graph on $n$ vertices. If $\tau : \cG^* \to [0,1]$ is $h$-local  then for any $ t \geq 0$, 
$$
\dP\PAR{    \sum_{v \in V(G) } \tau ( H, v) - \dE  \sum_{v \in V(G)  }  \tau ( H, v)  \geq  n t } \leq   \exp\PAR{ - \frac{ n    t^{2}}{  2 M^2 _h (G) } }.
$$
\end{lemma}

\begin{proof}
We may assume that the vertex $V$ of $G$ is $\{1, \ldots, n\}$. Let $A$ be the adjacency matrix of $H$. For $2 \leq k \leq n$, we define the vector $X_k = (A_{k\ell})_{1 \leq \ell \leq k-1} \in \cX_k = \{0,1\}^{k-1}$.  The graph $H$ can be recovered from $X = (X_2, \cdots, X_n) \in \cX = \times_{k=2}^n \cX_k$. Moreover, for some functions $F, F_v : \cX \to [0,1]$, $\tau(G,v) = F_v(X)$ and 
$$
F(X) = \sum_{v =1} ^n F_v(X).
$$
Let $X,X' \in \cX$ and assume that $X'_\ell = X_\ell$ unless $k = \ell$ for some $1 \leq k \leq n$. Then, since $\tau$ is $h$-local, we have $F_v (X)  = F_v (X')$ unless $v$ is within graph distance (in $G$) at most $h$ from vertex $k$. By definition, there are $N_h ( G,k)$ such vertices. It follows that 
$$
| F(X) - F(X') | \leq N_h ( G,k).
$$
It remains to apply Azuma-Hoeffding's inequality. 
\end{proof}

\subsection{Proof of Theorem \ref{th:percdeloc}}

\noindent{\em {Proof of (i), step one : control of the resolvent. }}
The beginning of the argument repeats the proof of Theorem \ref{th:perclocallaw}, we simply replace the Lebesgue measure $\ell$ by $\dP \otimes \ell$.  Let $S = [-d,d]$, $B$ be the adjacency operator of $\perc(\Gamma,p)$, $G_o  (z) = \langle e _o , ( B - z I) ^{-1} e_o \rangle$ and $h_{\eta} (\lambda) = \frac{1}{\pi} \Im G_o ( \lambda + i \eta)$.  Let $f$ be the random density of the absolutely continuous part of $\mu_{o} = \mu_{\perc (\Gamma,p)}^{e_o}$. We have that $\dP\otimes \ell$-a.s. $f (\lambda) = \lim_{\eta \to 0} h_\eta(\lambda)$. By assumption, $m  = \dE  \int_S f (\lambda ) d\lambda =  \dE (\mu_{o} )_{ac} (\dR)$ is positive 
and for any $\veps >0$, we can find a pair $(a,b)$ of positive numbers satisfying, 
$$
  \int_S \dE f (\lambda ) \IND ( a \leq f (\lambda ) \leq b) d\lambda  \geq m -  \veps/2.  
$$
Arguing as in Theorem \ref{th:perclocallaw}, there exist positive constants $a',b',q,\eta_0$ such that, 
if
 $$
\bar h_{\eta} (\lambda ) = \dE \left[  h_\eta( \lambda) \IND ( a' \leq  h_{t}( \lambda)  \leq b', \hbox{ for all } t \in [0,1])\right],$$
 the Borel set 
$
K_1 = \{ \lambda  \in S  : \bar h_\eta (\lambda)  > q  \hbox{ for all } \eta \in [0,1] \},
$
satisfies for any $0 \leq \eta \leq \eta_0$, 
\begin{equation}\label{eq:hetal2}
 \int_{K_1}\bar h_{\eta} (\lambda )  d\lambda  \geq m -   \veps.  
\end{equation}

We now use our concentration lemma to deduce from the above inequality an inequality satisfied by the resolvent of  $\perc(G,p)$. Without loss of generality, we may assume that $\delta$ given by \eqref{eq:defdmin} is smaller than $(a' \wedge q ) /4$ and that $\eta_1 = 20 d   \log (2h)  / h \leq \eta_0$. If $A$ is the adjacency matrix of $\perc(G,p)$ and $v \in V(G)$, we set $ \tilde h_{\eta,v} (\lambda) = \frac{1}{\pi} \Im (A - z ) ^{-1}_{vv}$, with  $z = \lambda + i \eta$.  We also set 
$\delta_0 =  1/  ( \zeta  h  \pi d)$  and  
$$\tau_{\eta,v} (\lambda)=  \frac{1}{\pi} \Im (A_v - z ) ^{-1}_{vv} \IND_{\cE_v(\lambda)},$$ 
where $A_{v}$ is the adjacency operator of the graph $(G,v)_h$ and $\cE_v(\lambda)$ denotes the event,
$$
\cE_v (\lambda) = \BRA{ a '- \delta_0 \leq \frac{1}{\pi} \Im (A_v - (\lambda + i t ) ) ^{-1}_{vv}  \leq b' + \delta_0, \hbox{ for all } t \in [\eta_1,1]}.
$$
First, if $v \notin B(h)$, then by  Corollary \ref{cor:jackson}, if $ \eta_1 \leq \eta \leq 1$,
$
\dE  \tau_{\eta,v} (\lambda)  \geq \bar h_{\eta} (\lambda )  - \delta_0.
$
Note that $(G,v) \mapsto \tau_{\eta,v} (\lambda)$ is a $h$-local functional in the sense defined above Lemma \ref{th:stabloc} and it is bounded by $1/ (\pi\eta_1) \leq (2h) / ( \zeta d \pi) $. We deduce from this lemma that if, $t = 2 \delta / (\zeta \pi d)$ and $\lambda \in \dR$,
\begin{equation}\label{eq:tauveta}
\frac 1 n \sum_{v=1}^n  \tau_{\eta,v} (\lambda)    \geq  \bar h_{\eta} (\lambda )  - \delta_0 -    \frac{2h}{ \zeta  \pi d} \frac{B(h)}{n}  -  t \geq \bar h_{\eta} (\lambda )  - \frac{5\delta}{ \zeta \pi d} ,
\end{equation}
with probability at least $1 - \exp ( -     n   \delta^{2} /  ( 2  h ^2  M_h ^2 (G) )  )$. As a consequence of Corollary \ref{cor:jackson}, if \eqref{eq:tauveta} holds, then
\begin{equation*}\label{eq:tauveta2}
\frac 1 n \sum_{v=1}^n  \tilde h_{\eta,v} (\lambda) \IND_{\tilde \cE_v(\lambda)}   \geq \bar h_{\eta} (\lambda ) -   \frac{\delta}{  \pi d},
\end{equation*}
where $\tilde \cE_v (\lambda)\supset \cE_v(\lambda)$ denotes the event,
$$ \tilde \cE_v(\lambda)= \BRA{ a' - 2\delta_0 \leq \tilde h_{t,v} (\lambda) \leq b' + 2 \delta_0,  \hbox{ for all } t \in [\eta_1,1]}. $$
We may use a net argument  as in Theorem \ref{th:perclocallaw}. We consider a $\delta \eta_1^2 $-net of $(8 d +2)/( \delta \eta_1^2 ) \leq h^2 / \delta$ points on boundary of the rectangle $R = \{ z :  \eta_1 \leq \Im ( z) \leq 1 , \Re (z) \in [-2d,2d] \}$. From the union bound and the maximum principle,
\begin{equation}\label{eq:supconv2}
\inf_{(\lambda + i \eta) \in R} \frac 1 n \sum_{v=1}^n  \tilde h_{\eta,v} (\lambda) \IND_{\tilde \cE'_v(\lambda)} - \bar h_{\eta} (\lambda )     \geq - \frac{3\delta}{4 \pi d} - \frac{\delta}{\pi} \geq - \delta,
\end{equation}
with probability  at least $1 - h ^2 \delta^{-1} \exp ( -     n   \delta^{2} /  ( 2  h ^2  M_h ^2 (G))  )$. In the above expression, $\tilde \cE'_v(\lambda)$ is defined as $\tilde \cE_v(\lambda)$ with $2 \delta_0$ replaced by $2\delta_0 + \delta / \pi$. We set $a'' = a' - 2\delta_0 - \delta / \pi$ and $b'' = b + 2\delta_0 + \delta / \pi$.

\vline

\noindent{\em {Proof of (i), step two : from resolvent to eigenvectors.}}
We may now use the above inequality to find delocalized eigenvectors in the finite graph $\perc(G,p)$. For $\lambda \in \dR$, let $$V(\lambda) = \{ v : \tilde \cE'_v (\lambda)  \hbox{ holds} \} = \{ v : a'' \leq \tilde  h_{\eta,v} (\lambda) \leq b'' \hbox{ for all } \eta_1 \leq \eta \leq 1 \}.$$ If \eqref{eq:supconv2} holds and $\lambda \in K_1$, then $| V(\lambda) | \geq (q - \delta) n  / b'' \geq  (q  / ( 2b'')) n $ with $q$ introduced above \eqref{eq:hetal2}. Notably, if \eqref{eq:supconv2} holds, by Corollary \ref{cor:psik}, we find that for any $\lambda \in K_1$, if $v \in V(\lambda)$, $\eta = c_0 (\log h )/ h$ and $I = [ \lambda - \eta , \lambda + \eta ]$,
\begin{equation}\label{eq:finalest}
c_1 \eta \leq \sum_{k \in \Lambda_I} \psi_k ( v) ^2 \leq c_2 \eta,
\end{equation}
for some positive constants $c_0 , c_1,c_2$. 
We sum over all $v \in V(h)$ and set $c_3 =  c_1 q  /  ( 2b'') $, we deduce that 
\begin{equation}\label{eq:lblambdaI}
 c_3  \eta  n \leq  c_1 \eta | V(\lambda) | \leq  \sum_{v \in V(\lambda) } \sum_{k \in \Lambda_I} \psi_k ( v) ^2 \leq | \Lambda_I |.
\end{equation}

On the other end, consider $L \subset K_1$ a maximal $2 \eta$-separated set, that is for any $\lambda \ne \lambda '$ in $L$, $|\lambda - \lambda' | \geq 2 \eta$ and $L$ has maximal cardinal. Then, by maximality, $K_1 \subset \cup_{\lambda \in L} ( \lambda-2 \eta,\lambda + 2 \eta)$, and it implies that $|L| \geq \ell (K_1) / ( 4\eta)$. Also, the set $\Lambda =  \cup_{\lambda \in L} \Lambda_{(\lambda- \eta, \lambda + \eta)}$ is a disjoint union.  Since $|\Lambda | \leq n$, from the pigeon hole principle, we deduce that there exists a subset $L^* \subset L$ of cardinal at least $|L| / 2 \geq \ell ( K_1) / ( 8\eta)$ such that for all $\lambda \in L^*$,
$$
| \Lambda_{(\lambda- \eta, \lambda + \eta)} | \leq  \frac{2 n }{ |L| } \leq \frac{ 8 }{\ell (K_1) } \eta n.
$$

We now prove that if $\lambda \in L^*$, then a positive proportion of the eigenvectors in $\Lambda_{(\lambda- \eta, \lambda + \eta)}$ have a positive proportion of their norm supported on $V(\lambda)$. To this end we apply the inequality \eqref{eq:finalest} to each  $\Lambda_{(\lambda- \eta, \lambda + \eta)}$ with $\lambda \in L^* $. We find that if \eqref{eq:supconv2} holds and $I = (\lambda - \eta, \lambda + \eta)$ then 
$$
\frac{1}{|\Lambda_I|} \sum_{ v \in V(\lambda)} \sum_{k \in \Lambda_I} \psi_k ( v) ^2 \geq c_1 \frac{\eta |V(\lambda)|}{|\Lambda_I|} \geq c_4,
$$
where $c_4 = c_1 q  \ell ( K_1) / ( 2^4 b'')$. For $k \in \Lambda_I$, let $x_k = \sum_{v \in V(\lambda)}  \psi_k ( v) ^2$. For $0 < t < 1$, if $\Lambda_I(t) = \{ k \in \Lambda_I : x_k > t\}$, we observe that, since $x_k \leq 1$,
$$
c_4 |\Lambda_I|  \leq  \sum_{k \in \Lambda_I} x_k  \leq  | \Lambda_I(t)|   + t (  |\Lambda_I| - |\Lambda_I(t) | ) . 
$$
We deduce that $\Lambda_I ( t) \geq |\Lambda_I| ( c_4 - t) / ( 1 - t)$. If $0 < t < c_4$, $\Lambda_I (t)$ is a positive proportion of $\Lambda_I$ and from \eqref{eq:lblambdaI},
$$
\ABS{  \bigcup_{\lambda \in L^*} \Lambda_{(\lambda - \eta, \lambda + \eta)} (t) } \geq \frac{1 - c_4}{1 -t} |L^* |  c_3  \eta  n \geq \alpha n, 
$$
with $\alpha = ( (1 - c_4)/(1 -t)) ( \ell ( K_1) / 8) c_3 $.

Finally, if $k \in \Lambda_{[\lambda - \eta, \lambda + \eta]}$ then \eqref{eq:finalest} implies that for any $v \in V(\lambda)$,
$
\psi_k ^2 ( v) \leq c_2 \eta. 
$ It thus concludes the proof of the first part of Theorem \ref{th:percdeloc} with $\rho = t$ and a new constant $c_1 = \sqrt{ c_0 c_2}$.

\vline

\noindent{\em {Proof of (ii).}} Let $\alpha = \rho = 1 - \sqrt \beta$, $\beta = 4 \pi ( 1 - m')$ and $m' = m - \veps - \delta$.  By Corollary \ref{cor:psik} (see \eqref{eq:lbst}), it is sufficient to prove (up to adjusting the constant $c$) that,  if \eqref{eq:supconv2} holds, there are at least $\alpha n$ eigenvalues $\lambda_k$, such that there exists a set $V_k$ and a real $l_k$ with $|l_k - \lambda_k | \leq c \eta$, 
$
y_k = \sum_{v \in V_k} \psi_k (v)^2 \geq \rho
$
and for all $v \in V_k$, $\tilde h_{\eta,v} ( l_k) \leq c$.

First, from \eqref{eq:hetal2}, if \eqref{eq:supconv2} holds, we have 
$$
 \int_{K_1} \frac 1 n  \sum_{v \in V(\lambda)} \tilde h_{\eta,v} ( \lambda)  d \lambda \geq  m'.
$$

We start with a regularization of the sets $V(\lambda)$. To this end, we consider the open set $U = \bigcup_{\lambda \in L} (\lambda - 3 \eta,\lambda + 3 \eta) \supset \bar K_1$ with $L$ as above. We can find a finite partition $(P_l)_l$ of $U$,  $U = \cup_{l} P_l$, such that $P_l$ is an interval of length at most $\pi \eta^2$ and no eigenvalue lies on the boundary $\partial P_l$.  For each $l$, we consider an element $x_l \in P_l$ and define, $V_l = \{ v \in V :   \tilde h_{\eta,v} ( x_l) \leq b''+1 \}$.  Then, since $ \ABS{\tilde h_{\eta,v} (x) -\tilde h_{\eta,v} (y) } \leq |x - y|\eta^{-2} \pi ^{-1}$,  we find for all $\lambda \in P_l$, $V(\lambda) \subset V_l$ and 
$$
\sum_{l} \frac 1 n \sum_{v \in V_l} \int_{P_l} \tilde h_{\eta,v} ( \lambda)  d \lambda \geq  m'.
$$
Since $\int  \frac 1 n \sum_{v =1} ^n    \tilde h_{\eta,v} ( \lambda)  d \lambda =1$, we get 
$$ \sum_{l} \frac 1 n \sum_{v \notin V_l} \int_{P_l} \tilde h_{\eta,v} ( \lambda)  d \lambda \leq 1 - m'.$$
Now, we observe that $\sum_{v \notin V_l }   \tilde h_{\eta,v} ( \lambda) $ is equal to $\Im (g_{\mu_l} (z))/ \pi$ with $z = \lambda + i \eta$, $\mu_l = \sum_{k}  ( 1 - y_{k,l}) \delta_{\lambda_k}$ and
$$
y_{k,l} = \sum_{v \in V_l} \psi_k (v)^2 = 1 - \sum_{v \notin V_l} \psi_k (v)^2.
$$
If $k \in \Lambda_U$, then $ k \in \Lambda_{P_l}$ for a unique $l$, and we set $y_k = y_{k,l}$.  From Lemma \ref{le:coaire} and the assumption $\mu_l (\partial P_l) = 0$, we find
$$
\sum_{k  \in \Lambda_U} ( 1- y_k) =\sum_l  \sum_{k  \in \Lambda_{P_l} } ( 1- y_{k,l}) = \sum_{l } \mu_l (P_l)  \leq n   2 \pi ( 1 - m').
$$
 Similarly,  if \eqref{eq:supconv2} holds, 
$$
 \int_{K^c_1} \frac 1 n \sum_{v=1} ^n \tilde h_{\eta,v} ( \lambda)  d \lambda \leq 1 - m'.
$$
Then, we apply Lemma \ref{le:coaire} to the open set $(\bar K_1 )^c \subset K_1 ^c$. Since $U^c \subset   (\bar K_1 )^c $, we find
$$
| \Lambda_{U^c} | \leq n 2 \pi ( 1 - m').
$$
Hence, we have checked that, if \eqref{eq:supconv2} holds, 
$$
\sum_{k \in \Lambda_U} y_k \geq n ( 1 - \beta).
$$
For $0< t < 1$, if $n_t = \{ k \in \Lambda_U :  y_k > t \}$, we have $n ( 1 - \beta) \leq ( n -n_t) t + n_t$. Hence, 
$$
n_t \geq \frac{ 1 - \beta - t }{1 - t} n = \alpha n. 
$$
We choose $t  = \rho = 1 -  \sqrt \beta$, we get $\alpha = 1 - \sqrt \beta$. It concludes the proof of the theorem.

\bibliographystyle{plain}
\bibliography{bib}

\bigskip
\noindent
 Charles Bordenave \\
 Institut de Math\'ematiques de Toulouse. CNRS and University of Toulouse. \\
118 route de Narbonne. 31062 Toulouse cedex 09.  France. \\
\noindent
{E-mail:} {\tt bordenave@math.univ-toulouse.fr} \\
\noindent
\url{http://www.math.univ-toulouse.fr/~bordenave}

\end{document}